\documentclass[10pt,a4paper,reqno]{amsart}
\usepackage{acronym}
\usepackage{amsfonts}
\usepackage{amsmath}
\usepackage{amssymb}
\usepackage{amsthm}
\usepackage{bm}
\usepackage{braket}
\usepackage{cite}
\usepackage{color}
\usepackage{geometry}
\usepackage{graphicx}
\usepackage{hyperref}
\usepackage{mathrsfs}
\usepackage{mathtools}
\usepackage{setspace}
\usepackage{stmaryrd}
\usepackage{subfigure}
\usepackage{tikz}

\usepackage{epsfig}
\usepackage{setspace}
\usepackage{booktabs}
\usepackage{threeparttable}
\usepackage{diagbox}
\usepackage{epstopdf}

\newgeometry{top=2cm,bottom=2cm,outer=1.5cm,inner=1.5cm}
\newcommand{\bse}{\begin{subequations}}
\newcommand{\ese}{\end{subequations}}

\newtheorem{theorem}{Theorem}

\newtheorem{lemma}[theorem]{Lemma}

\newcommand{\R}{\mathbb R}
\newcommand{\Z}{\mathbb Z}
\newcommand{\C}{\mathbb C}

\usepackage{amssymb}
\usepackage[]{amsmath}
\usepackage{graphics}
\usepackage{graphicx}
\usepackage{amssymb}
\usepackage[]{amsmath}
\usepackage{amsthm}
\usepackage{setspace}
\usepackage{epsfig}

\usepackage{color}
\usepackage{pifont,bm}
\usepackage{mathrsfs}
\usepackage{booktabs}
\usepackage{diagbox}
\usepackage{makecell}
\newtheorem{proposition}{Proposition}

\numberwithin{equation}{section}

\title[High-order soliton solutions and their dynamics in the  inhomogeneous variable coefficients Hirota equation]{High-order soliton solutions and their dynamics in the  inhomogeneous variable coefficients Hirota equation}
\author{Huijuan Zhou}
\address[HZ]{School of Mathematical Sciences, Shanghai Key Laboratory of Pure Mathematics and Mathematical Practice, and Shanghai Key Laboratory of Trustworthy Computing \\
East China Normal University \\ Shanghai 200241 \\ People's Republic of China}
\author{Yong Chen}
\address[YC]{School of Mathematical Sciences, Shanghai Key Laboratory of Pure Mathematics and Mathematical Practice, and Shanghai Key Laboratory of Trustworthy Computing \\
East China Normal University \\ Shanghai 200241 \\ People's Republic of China}
\address[YC]{College of Mathematics and Systems Science \\ Shandong University of Science and Technology \\ Qingdao 266590 \\ People's Republic of China}
\address[YC]{Department of Physics \\ Zhejiang Normal University \\ Jinhua 321004 \\ People's Republic of China}
\email{ychen@sei.ecnu.edu.cn}

\begin{document}
\par
\begin{abstract}
A series of new soliton solutions are presented for the inhomogeneous variable coefficient Hirota equation by using the Riemann Hilbert method and transformation relationship. First, through a standard dressing procedure, the N-soliton matrix associated with the simple zeros in the Riemann Hilbert problem for the Hirota equation is constructed. Then the N-soliton matrix of the inhomogeneous variable coefficient Hirota equation can be obtained by a special transformation relationship from the N-soliton matrix of the Hirota equation. Next, using the generalized Darboux transformation, the high-order soliton solutions corresponding to the elementary high-order zeros in the Riemann Hilbert problem for the Hirota equation can be derived. Similarly, employing the transformation relationship mentioned above can lead to the high-order soliton solutions of the inhomogeneous variable coefficient Hirota equation. In addition, the collision dynamics of Hirota and inhomogeneous variable coefficient Hirota equations are analyzed; the asymptotic behaviors for  multi-solitons and long-term asymptotic estimates for the high-order one-soliton of the Hirota equation are concretely calculated. Most notably, by analyzing the dynamics of the multi-solitons and high-order solitons of the inhomogeneous variable coefficient Hirota equation, we discover numerous new waveforms such as heart-shaped periodic wave solutions, O-shaped periodic wave solutions etc. that have never been reported before, which are crucial in theory and practice.
\end{abstract}

\maketitle

\section{Introduction}

As we all know, the nonlinear Schr\"{o}dinger (NLS) equation describes a plane self-focusing and one-dimensional self-modulation of waves in nonlinear dispersive media, which has various applications in a wide range of physical systems such as water waves \cite{be1967,za1968}, nonlinear optics \cite{ha19731,ha19732}, solid-state physics and plasma physics \cite{ha1972}. However, several phenomena observed in the experiment cannot be justified by the NLS equation. As the light pulse becomes shorter, they require more energy to become solitons \cite{post-1989-ap,lib-2005}. In this case, some additional effects have become significant. A modified NLS equation
\begin{equation}\label{hirota}
iu_{Z}+\alpha(u_{TT}+2|u|^{2}u)+i\beta (u_{TTT}+6|u|^{2}u_{T})=0, \\
\end{equation}
which is called Hirota equation, can be used to describe the propagation of subpicosecond or femtosecond optical pulse in fibers thanks to it taking into account higher-order dispersion and time-delay corrections to the cubic nonlinearity. In this equation, $T$ is the propagation variable and $Z$ is the retarded time variable in a moving frame while $u$ is the envelope of the wave field. The two terms in \eqref{hirota} with a real coefficient $\beta$ represent the third-order dispersion $u_{TTT}$ and a time-delay correction to the cubic term $|u|^{2}u_{T}$, respectively. Due to a fine balance between its linear dispersive and nonlinear collapsing terms, some exact solitons solutions of Hirota equation have been obtained by many author\cite{vc4,vc5,vc6,vc7,vc8,vc9,vc10,vc11,vc13,vc14}. When $\alpha=1$, $\beta=0$,  Eq.\eqref{hirota} reduces to the NLS equation. In another case of $\alpha=0$, $\beta=1$, the Hirota equation reduces to modified Korteweg-de Vries equation.

Note that these investigations of optical solitons or solitary waves have been focused mainly on homogeneous fibers. However, considering long-distance communication and manufacturing problems in the realistic fiber transmission lines, the inhomogeneous variable coefficient Hirota  (IVC-Hirota) equation \cite{kodama-1985} 
\begin{equation}\label{vcheq}
iq_{z}+\alpha_{1}(z) q_{t t}-\frac{1}{3 \delta} i \alpha_{1}(z) q_{t t t}+\delta \alpha_{4}(z) q|q|^{2}-i\alpha_{4}(z)|q|^{2} q_{t}-i\alpha_{6}(z)q=0,
\end{equation}
was investigated. 
Here $$
\alpha_{6}(z)=\frac{\alpha_{1, z} \alpha_{4}-\alpha_{1} \alpha_{4, z}}{2 \alpha_{1} \alpha_{4}},
$$
$\delta$ is a real number, $\alpha_{1}(z)$ and $\alpha_{4}(z)$ are dispersion and nonlinear effects respectively.  Dispersion broadens the waveform and nonlinear effects narrow it. Under certain conditions, the two effects reach a balance and maintain waveform stability.
We should point out that Eq.\eqref{vcheq} is an integrable equation. Study of  Eq.\eqref{vcheq} is of great interest due to its wide range of applications \cite{vc29,dai-2006-jpa,P-SAM-2010,tao-JNMP-2013,Rajan-2015,Gao-2017,Yang-RC-2021}. Its use is not only restricted to optical pulse propagation in inhomogeneous fiber media, but also to the core of dispersion-managed solitons and combined managed solitons. 
To our knowledge, the study of high-order soliton for the IVC-Hirota equation has not been widespread.

In this paper, an IVC-Hirota equation that is widely used in optics is studied by the Riemann-Hilbert (RH) method and transformation relationship. There is a transformation relationship, which maps IVC-Hirota to the Hirota equation. Thanks to the transformation, one can obtain many solutions of the IVC-Hirota equation from the known solutions of the Hirota equation. Specifically, as for the Hirota equation, through a standard dressing procedure, we can find the soliton matrix for the nonregular RH problem with simple zeros. Then combined with generalized Darboux transformation (gDT), soliton matrix for elementary high-order zeros in the RH problem are constructed. RH problem provides a feasible and strict method for studying the long-term asymptotic behavior of integrable equations \cite{vc86,pwq-2019,feg}. Furthermore, the influence of free parameter $(\alpha, \beta)$ in soliton solutions of general Hirota equation on soliton propagation, collision dynamics along with the asymptotic behavior for the two solitons and longtime asymptotic estimations for the high-order one soliton are concretely analyzed. The propagation direction, velocity, width and other physical quantities of solitons can be modulated by adjusting the free parameters of the general Hirota equation. In addition, using a special transformation, we can obtain explicit expressions of multi-solitons and high-order soliton for the IVC-Hirota equation by the solutions of the Hirota equation. We design abundant new types of multi-solitons and high-order solitons of the IVC-Hirota equation through analysis of the explicit expression of solutions. Such as: heart-shaped periodic wave solutions and O-shaped periodic wave solutions.
The dynamics analysis of these solutions are useful in observing and design of fiber optic in femtosecond fiber laser systems or in optical communication links with distributed dispersion and nonlinearity management.

This paper is organized as follows. In section 2, the matrix RH problem is formulated. In section 3, the N-soliton formula for Hirota and IVC-Hirota equation is derived by considering the simple zeros in the RH problem and some exact solutions are constructed. In section 4, we first give the high order N-soliton formula for the Hirota equation, which corresponds to the elementary zeros in the RH problem.  Then the high order N-soliton formula for IVC-Hirota is also constructed.
In section 5, the dynamics of high-order solitons are given in the IVC-Hirota equation. The final section is devoted to conclusion and discussion.

\section{Inverse scattering theory for the Hirota equation}

In this section, we study the scattering and inverse scattering problem for the Hirota equation \eqref{hirota}. The Hirota equation  can be constructed by the compatibility condition of the following spectral problem\cite{sa1991}:
\begin{equation}\label{b}
\begin{split}
Y_{T}=UY,\\
Y_{Z}=VY,
\end{split}
\end{equation}
where
\begin{equation}\notag
\begin{split}
U=-i\zeta\Lambda+Q, \quad
V=(-4i\beta\zeta^3-2i\alpha\zeta^{2})\Lambda+V_{1},\quad
Q=\left(\begin{matrix}
0&u\\
-u^{*}&0\end{matrix}
\right), \quad
\Lambda=\left(\begin{matrix}
1&0\\
0&-1\end{matrix}
\right),
\end{split}
\end{equation}
\begin{equation}\notag
V_{1}=\zeta^{2}\left(\begin{matrix}
0&4\beta u \\
-4\beta u^{*} &0\end{matrix}
\right)+\zeta\left(\begin{matrix}
2i\beta|u|^2&2i\beta u_{T}+2\alpha u\\
2i\beta u_{T}^{*}-2\alpha u^{*}&-2i\beta|u|^2\end{matrix}
\right)+\left(\begin{matrix}
i\alpha|u|^2+\beta(uu_{T}^{*}-u_{T}u^{*})&i\alpha u_{T}-\beta(2|u|^2u+u_{TT}) \\
i\alpha u_{T}^{*}+\beta(2|u|^2u^{*}+u_{TT}^{*}) &-i\alpha|u|^2-\beta(uu_{T}^{*}-u_{T}u^{*})\end{matrix}
\right),
\end{equation}
 $u=u(T,Z)$ is potential function, $\zeta$ is a spectral parameter, $Y(T,Z,\zeta)$ is a vector function, and the superscript $*$ represents complex conjugation. Supposing $u(T)=u(T,0)$ decays to zero sufficiently fast as $T \rightarrow \pm\infty$. For a prescribed initial condition $u(T,0)$, we seek the solution $u(T,Z)$ at any later time $Z$. That is, we solve an initial value problem for the Hirota equation.  

Notation
\begin{equation}\notag
E_{1}=e^{-i\zeta\Lambda T-(4i\beta\zeta^3+2i\alpha\zeta^{2})\Lambda Z},
\end{equation}
\begin{equation}\label{2.5}
J=YE_{1}^{-1},
\end{equation}
so that the new matrix function $J$ is $(T,Z)$-independent at infinity. Inserting \eqref{2.5} into \eqref{b}, the Lax pair \eqref{b} becomes
\begin{equation}\label{2.6}
\begin{split}
&J_{T}=-i\zeta[\Lambda,J]+QJ,\\
&J_{Z}=-(4i\beta\zeta^3+2i\alpha\zeta^{2})[\Lambda,J]+V_{1}J,
\end{split}
\end{equation}
where $[\Lambda,J]=\Lambda J-J\Lambda$ is the commutator. Notice that both matrices $Q$ and $V_{1}$ are anti-Hermitian, i.e.,
\begin{equation}\label{2.9}
Q^{\dag}=-Q, \ \ V_{1}^{\dag}=-V_{1},
\end{equation}
where the superscript $\dag$ represents the Hermitian of a matrix. In addition, their traces are both equal to zero, i.e., $trQ = trV_{1} = 0$.

Now the time $Z$ in the above notations to be considered as dummy variable. For the scattering problem, introduce matrix Jost solutions $J_{\pm}(T,\zeta)$ of \eqref{2.6} with the following asymptotic at large distances:
\begin{equation}\label{2.10}
J_{\pm}(T,\zeta)\rightarrow I, \quad T\rightarrow\pm\infty,
\end{equation}
where $I$ is a $2\times2$ unit matrix. Next, the analytical properties of Jost solutions $J_{\pm}(T,\zeta)$ will be delineated.  First, the notation $E(T,\zeta)=e^{-i\zeta \Lambda T}$, $\Phi\equiv J_{-}E$ and $\Psi \equiv J_{+}E$ are introduced. Notice that $Y_{\pm}(T,\zeta)$ satisfies the scattering equation \eqref{b}, i.e.,
\begin{equation}\label{2.19}
Y_{T}+i\zeta\Lambda Y=QY.
\end{equation}
Treating the $QY$ term in the above equation as an inhomogeneous term and noticing the solution to the homogeneous equation on its left side is $E$, then the equation  \eqref{2.19} can be turned into Volterra integral equations by using the method of variation of parameters as well as the boundary conditions \eqref{2.10}. These equations can be cast in terms of  $J_{\pm}$ as
\begin{equation}\label{2.20}
J_{-}(T,\zeta)=I+\int_{-\infty}^{T}e^{-i\zeta\Lambda(T-y)}Q(y)J_{-}(y,\zeta)e^{-i\zeta\Lambda(y-T)}dy,
\end{equation}
\begin{equation}\label{2.21}
J_{+}(T,\zeta)=I-\int_{T}^{\infty}e^{i\zeta\Lambda(y-T)}Q(y)J_{+}(y,\zeta)e^{i\zeta\Lambda(T-y)}dy.
\end{equation}

Thus, as long as the integrals on the right sides of the above Volterra equations converge, $J_{\pm}(T,\zeta)$ allow analytical continuations off the real axis $\zeta \in R$. The following proposition can easily be derived through the structure of the potential $Q$.
\begin{proposition}\label{p1}
 The first column of $J_{-}$ and the second column of $J_{+}$
can be analytically continued to the upper half plane $\zeta \in \C_{+}$ , while the second column of
$J_{-}$ and the first column of $J_{+}$ can be analytically continued to the lower half plane $\C_{-}$.
\end{proposition}
\begin{proof}
The integral equation \eqref{2.20} for the first column of $J_{-}$ , say $\left(\begin{array}{cc}
\varphi_{1} \\
\varphi_{2}
\end{array}\right)
$, is
\begin{equation}\label{2.22}
\varphi_{1}=1+\int_{-\infty}^{T}u(y)\varphi_{2}(y,\zeta)dy,
\end{equation}
\begin{equation}\label{2.23}
\varphi_{2}=-\int_{-\infty}^{T}u^{*}(y)\varphi_{1}(y,\zeta)e^ {2i\zeta(T-y)}dy.
\end{equation}
When $\zeta \in \C_{+}$, since $e^{2i\zeta (T-y)}$ in \eqref{2.23} is bounded, and $u(T)$ decays to zero sufficiently fast at large distances, both integrals in the above two equations converge. Thus the Jost solution $\left(\begin{array}{cc}
\varphi_{1} \\
\varphi_{2}
\end{array}\right)
$ can be analytically extended to $\C_{+}$. The analytic properties of the other Jost solutions $J_{+}$ can be obtained similarly.
\end{proof}

From Abel's identity, we find that $|J(T,\zeta)|$ is a constant for all $T$. Then using the boundary conditions \eqref{2.10}, we see that
\begin{equation}\label{2.13}
|J_{\pm}(T,\zeta)|=1,
\end{equation}
for all $(T,\zeta)$. Since $\Phi(T,\zeta)$ and $\Psi(T,\zeta)$ are both solutions of the linear equation \eqref{b}, they are linearly related by a scattering matrix $S(\zeta)$:
\begin{equation}\label{2.16}
\Phi(T,\zeta)=\Psi(T,\zeta)S(\zeta), \quad \zeta\in \R.
\end{equation}
i.e.,
\begin{equation}\label{2.17}
J_{-}=J_{+}ESE^{-1}, \quad \zeta\in \R,
\end{equation}
here $\R$ is the set of real numbers.

Because the potential $u(T,Z)$ can be reconstructed by using the scattering matrix $S(\zeta)$, so the analytical properties of $S(\zeta)$ need to be delineated first. If $(\Phi,\Psi)$ are  expressed as a collection of columns
\begin{equation}\notag
\Phi=(\phi_{1},\phi_{2}), \quad \Psi=(\psi_{1}, \psi_{2}),
\end{equation}
from Proposition 1 and $\Phi\equiv J_{-}E$ and $\Psi \equiv J_{+}E$, we have
\begin{equation}\notag
\Phi=(\phi_{1}^{+},\phi_{2}^{-}), \quad \Psi=(\psi_{1}^{-}, \psi_{2}^{+}),
\end{equation}
\begin{equation}\notag
\Phi^{-1}=\left(\begin{matrix}
\hat{\phi_{1}}^{-}\\
\hat{\phi_{2}}^{+}\end{matrix}
\right), \quad
\Psi^{-1}=\left(\begin{matrix}
\hat{\psi_{1}}^{+}\\
\hat{\psi_{2}}^{-}\end{matrix}
\right),
\end{equation}
where the superscripts $\pm$ indicate the half plane of analyticity for the underlying quantities. Since
\begin{equation}\notag
S=\Psi^{-1}\Phi=\left(\begin{matrix}
\hat{\psi_{1}}^{+}\\
\hat{\psi_{2}}^{-}\end{matrix}
\right)(\phi_{1}^{+},\phi_{2}^{-}),
\end{equation}
\begin{equation}\notag
S^{-1}=\Phi^{-1}\Psi=\left(\begin{matrix}
\hat{\phi_{1}}^{-}\\
\hat{\psi_{2}}^{+}\end{matrix}
\right)(\psi_{1}^{-},\psi_{2}^{+}),
\end{equation}
it is easy to see that scattering matrices $S$ and $S^{-1}$ have the following analyticity structures:
\begin{equation}\notag
S=\left(\begin{matrix}
s_{11}^{+},s_{12}\\s_{21},s_{22}^{-}\end{matrix}
\right), \quad
S^{-1}=\left(\begin{matrix}
\hat{s}_{11}^{-},\hat{s}_{12}\\ \hat{s}_{21},\hat{s}_{22}^{+}\end{matrix}
\right).
\end{equation}
The elements  without superscripts indicate that such  elements do not allow analytical extensions to $\C_{\pm}$ in general. From $S$ is a $2\times2$ matrix with unit determinant, we have
\begin{equation}\notag
\hat{s}_{11} = s_{22},\quad  \hat{s}_{22} =s_{11},\quad \hat{s}_{12}=-s_{12},\quad  \hat{s}_{21} =-s_{21}.
\end{equation}
Therefore, the analytic properties of $S^{-1}$ are apparent from the analytic properties of $S$.

In order to construct the RH problem, we define the Jost solutions
\begin{equation}\label{2.25}
P^{+}=(\phi_{1},\psi_{2})e^{i\zeta\Lambda T}=J_{-}H_{1}+J_{+}H_{2}
\end{equation}
are analytic in $\zeta \in \C_{+}$, and the Jost solutions
\begin{equation}\label{2.26}
(\psi_{1},\phi_{2})e^{i\zeta\Lambda T}=J_{+}H_{1}+J_{-}H_{2}
\end{equation}
are analytic in $\zeta \in \C_{-}$, here $H_{1}\equiv diag(1,0)$, $H_{2}\equiv diag(0,1)$.
In addition, from the Volterra integral equations \eqref{2.20}-\eqref{2.21}, we see that the large $\zeta$ asymptotics of these analytical functions are
\begin{equation}\label{2.28}
P^{+}(x,\zeta)\rightarrow I, \quad \zeta \in \C_{+}\rightarrow\infty,
\end{equation}
\begin{equation}\label{2.29}
(\psi_{1},\phi_{2})e^{i\zeta\Lambda T}\rightarrow I, \quad \zeta \in \C_{-}\rightarrow\infty.
\end{equation}
To obtain the analytic counterpart of $P^{+}$ in $\C_{-}$, we consider the adjoint scattering equation of \eqref{2.6}:
\begin{equation}\label{2.30}
K_{T}=-i\zeta[\Lambda,K]-KQ.
\end{equation}
Indeed, by utilizing the relation
\begin{equation}\label{2.31}
0=(JJ^{-1})_{T}=J_{T}J^{-1}+J(J^{-1})_{T}
\end{equation}
as well as the scattering equation \eqref{2.6}, we have
\begin{equation}\label{2.32}
J^{-1}_{T}=-i\zeta[\Lambda,J^{-1}]-J^{-1}Q,
\end{equation}
thus $J^{-1}(T,\zeta)$ satisfies the adjoint equation \eqref{2.30}. If we express $\Phi^{-1}$ and $\Psi^{-1}$ as a collection of rows
\begin{equation}\label{2.33}
\Phi^{-1}=\left(\begin{matrix}
\hat{\phi_{1}}\\
\hat{\phi_{2}}\end{matrix}
\right),\quad
\Psi^{-1}=\left(\begin{matrix}
\hat{\psi_{1}}\\
\hat{\psi_{2}}\end{matrix}
\right).
\end{equation}
Similarly, we can show that the adjoint Jost solutions
\begin{equation}\label{2.34}
P^{-}= e^{-i\zeta\Lambda x}\left(\begin{matrix}
\hat{\phi_{1}}\\
\hat{\psi_{2}}\end{matrix}
\right)=H_{1}J_{-}^{-1}+H_{2}J_{+}^{-1}
\end{equation}
are analytic in $\zeta \in \C_{-}$. In addition,
\begin{equation}\label{2.36}
P^{-}(T,\zeta)\rightarrow I, \quad \zeta \in C_{-}\rightarrow \infty.
\end{equation}

The anti-Hermitian property \eqref{2.9} of the potential matrix $Q$ gives rise to involution properties in the scattering matrix as well as in the Jost solutions. Indeed, by taking the Hermitian of the scattering equation \eqref{2.6} and utilizing the anti-Hermitian property of the potential matrix $Q^{\dag}=-Q$, we get
\begin{equation}\notag
J^{\dag}_{T} = -i\zeta^{*}[\wedge,J^{\dag}]-J^{\dag}Q.
\end{equation}
Thus $J^{\dag}_{\pm} (T,\zeta^{*})$ satisfy the adjoint scattering equation \eqref{2.30}.
However, $J^{-1}_{\pm} (T,\zeta)$ satisfies this adjoint equation as well. Consequently, $J^{\dag}_{\pm} (T,\zeta^{*})$ and $J^{-1}_{\pm} (T,\zeta)$  must be linearly dependent on each other. Recalling the boundary
conditions \eqref{2.10} of Jost solutions$J _{\pm}$, we further see that $J^{\dag}_{\pm} (T,\zeta^{*})$ and $J^{-1}_{\pm} (T,\zeta)$  have the
same boundary conditions at $T \rightarrow \pm\infty$, and hence they must be the same solutions of the
adjoint equation \eqref{2.30}, i.e. $J^{\dag}_{\pm}(\zeta^{*}) = J^{-1}_{\pm}(\zeta)$.
From this involution property as well as the definitions \eqref{2.25} and \eqref{2.34} for $P^{\pm}$, we see that the analytic solutions $P^{\pm}$ satisfy the involution property as well:
\begin{equation}\label{2.57}
(P^{+})^{\dag}(\zeta^{*}) = P^{-}(\zeta).
\end{equation}
In addition, in view of the scattering relation \eqref{2.17} between $J_{+}$ and $J_{-}$, we see that $S$ also satisfies the involution property:
\begin{equation}\label{2.58}
S^{\dag}(\zeta^{*}) = S^{-1}(\zeta).
\end{equation}

\subsection{Matrix Riemann-Hilbert problem}

On the real line, using \eqref{2.17}, \eqref{2.25} and \eqref{2.34}, we can easily get
\begin{equation}\label{2.44}
P^{-}(T,\zeta)P^{+}(T,\zeta)= G(T,\zeta),\quad \zeta\in \mathbb{R},
\end{equation}
where
\begin{equation}\notag
G = E(H_{1} +H_{2}S)(H_{1}+S^{-1}H_{2})E^{-1}=E\left(\begin{matrix}
1 \ \ \,\hat{s}_{12}\\s_{21}\ \ 1\end{matrix}
\right)E^{-1}.
\end{equation}
Equation \eqref{2.44} forms a matrix RH problem. The normalization condition for this RH problem can be obtained from \eqref{2.28} and \eqref{2.36} as
\begin{equation}\label{2.46}
P^{\pm}(T,\zeta)\rightarrow I, \ \ \ \zeta\in \infty,
\end{equation}
which is the canonical normalization condition. If this RH problem can be
solved from the given scattering data $(s_{21} ,\widehat{s}_{12})$, then the potential $Q$ can be reconstructed from the asymptotic expansion of its solution at large $\zeta$. Indeed, recall that $P^{+}$ and $P^{-}$ are solutions of the scattering problem \eqref{2.6} and its adjoint problem \eqref{2.30}, respectively.

Recalling the definitions \eqref{2.25} and \eqref{2.34} of $P^{\pm}$ as well as the scattering relation \eqref{2.17}, we have that
\begin{equation}\label{2.52}
|P^{+}|= \widehat{s}_{22} = s_{11},\quad |P^{-}|= s_{22} = \widehat{s}_{11}.
\end{equation}
The RH problem \eqref{2.44} is called regular when $|P^{\pm}| \neq 0$. First, the solution of the regular RH problem is considered. Namely, $\widehat{s}_{22} = s_{11}= s_{22} = \widehat{s}_{11}\neq0$ in their respective planes of analyticity. Under the canonical normalization condition \eqref{2.46}, the solution to this regular RH problem is unique\cite{yjks}. This unique solution to the regular matrix RH problem \eqref{2.44} defies explicit expressions. Its formal solution, however, can be given in terms of a Fredholm integral equation.

To use the Plemelj-Sokhotski formula on the regular RH problem \eqref{2.44}, first rewrite the \eqref{2.44} as
\[
\left(P^{+}\right)^{-1}(\zeta)-P^{-}(\zeta)=\widehat{G}(\zeta)\left(P^{+}\right)^{-1}(\zeta), \quad \zeta \in \mathbb{R}
\]
where
\[
\widehat{G}=I-G=-E\left(\begin{array}{cc}
0 & \hat{s}_{12} \\
s_{21} & 0
\end{array}\right) E^{-1},
\]
$\left(P^{+}\right)^{-1}(\zeta)$ is analytic in $\mathbb{C}_{+},$ and $P^{-}(\zeta)$ is analytic in $\mathbb{C}_{-}$. Applying the Plemelj-Sokhotski formula and utilizing the canonical boundary conditions \eqref{2.46}, the solution to the
regular RH problem \eqref{2.44} is provided by the following integral equation:
\[
\left(P^{+}\right)^{-1}(\zeta)=I+\frac{1}{2 \pi i} \int_{-\infty}^{\infty} \frac{\widehat{G}(\xi)\left(P^{+}\right)^{-1}(\xi)}{\xi-\zeta} d \xi, \quad \zeta \in \mathbb{C}_{+}.
\]

In the more general case, the RH problem \eqref{2.44} is not regular; i.e., $|P^{+}(\zeta)|$ and $|P^{-}(\zeta)|$ can be zero at certain discrete locations $\zeta_{k} \in \mathbb{C}_{+}$ and $\bar{\zeta}_{k} \in \mathbb{C}_{-}, 1 \leq k \leq N$, where $N$ is the number of these zeros. In view of \eqref{2.52}, we see that $\left(\zeta_{k}, \bar{\zeta}_{k}\right)$ are zeros of the scattering coefficients $\hat{s}_{22}(\zeta)$ and $s_{22}(\zeta)$. Due to the involution property \eqref{2.58}, we have the involution relation
\begin{equation}\label{2.66}
\bar{\zeta}_{k}=\zeta_{k}^{*}.
\end{equation}
For simplicity, we assume that all zeros $\left\{\left(\zeta_{k}, \bar{\zeta}_{k}\right), k=1, \ldots, N\right\}$ are simple zeros of $\left(\hat{s}_{22}, s_{22}\right)$ which is the generic case. In this case, the kernels of $P^{+}\left(\zeta_{k}\right)$ and $P^{-}\left(\bar{\zeta}_{k}\right)$ contain only a single column vector $\left|v_{k}\right\rangle$ and row vector $\left\langle \overline{v}_{k}\right|$, respectively. I.e.,
\begin{equation}\label{2.67}
P^{+}\left(\zeta_{k}\right) \left|v_{k}\right\rangle=0, \quad \left\langle \overline{v}_{k}\right| P^{-}\left(\bar{\zeta}_{k}\right)=0, \quad 1 \leq k \leq N.
\end{equation}
Taking the Hermitian of the first equation in \eqref{2.67} and utilizing the involution properties \eqref{2.57} and \eqref{2.66}, 
\begin{equation}\label{eqvk}
\left|v_{k}\right\rangle^{\dagger} P^{-}\left(\bar{\zeta}_{k}\right)=0
\end{equation} can be got. Then comparing equation \eqref{eqvk} with the second equation in \eqref{2.67}, we know that eigenvectors $\left(\left|v_{k}\right\rangle, \left\langle \overline{v}_{k}\right|\right)$
satisfy the involution property $
\left\langle \overline{v}_{k}\right|=\left|v_{k}\right\rangle^{\dagger}$.
Vectors $\left|v_{k}\right\rangle$ and $\left\langle \overline{v}_{k}\right|$ are $T$ dependent, taking the $T$ derivative to the eq.\eqref{2.67} and recalling that $P^{+}$ satisfies the scattering equation \eqref{2.6}, we have
\[
P^{+}\left(\zeta_{k} ; T\right)\left(\frac{d \left|v_{k}\right\rangle}{d T}+i \zeta_{k} \Lambda \left|v_{k}\right\rangle\right)=0.
\]
Due to our assumption, the only vector in the kernel of $P^{+}\left(\zeta_{k} ; T\right)$ is $\left|v_{k}\right\rangle$. Thus
\[
\frac{d \left|v_{k}\right\rangle}{d T}+i \zeta_{k} \Lambda \left|v_{k}\right\rangle=\alpha_{k}(x) \left|v_{k}\right\rangle,
\]
where $\alpha_{k}(T)$ is a scalar function. The solution to the above equation is
\[
\left|v_{k}(T)\right\rangle=e^{-i \zeta_{k} \Lambda T} \left|v_{k0}\right\rangle e^{\int_{T_{0}}^{T} \alpha_{k}(y) dy},
\]
where $\left|v_{k0}\right\rangle=\left|v_{k}(T)\right\rangle|_{T=0}$. Without loss of generality, we take $\alpha_{k}=0$ and write the solution $\left|v_{k}(T)\right\rangle$ as
\begin{equation}\label{2.99}
\left|v_{k}(T)\right\rangle=e^{-i\zeta_{k} \Lambda T}\left|v_{k0}\right\rangle.
\end{equation}
Following similar calculations for $\bar{v}_{k},$ we readily get
\[
\left\langle \overline{v}_{k}(T)\right|=\left\langle \bar{v}_{k0}\right| e^{i \zeta_{k} \Lambda T}.
\]
These two equations give the simple $T$ dependence of vectors $\left|v_{k}(T)\right\rangle$ and $\left\langle \bar{v}_{k}(T)\right|$.
The zeros $\left\{\left(\zeta_{k}, \bar{\zeta}_{k}\right)\right\}$ of $|P^{\pm}(\zeta)|$ as well as vectors ${\left|v_{k}\right\rangle, \left\langle \overline{v}_{k}\right|}$ in the kernels of $P^{+}\left(\zeta_{k}\right)$ and $P^{-}\left(\bar{\zeta}_{k}\right)$ constitute the discrete scattering data which is also needed to solve the general RH problem \eqref{2.44}.

Now introduce a matrix function that could remove all the zeros of this RH problem. For this purpose, first introduce the rational matrix function:
\[
\Gamma_{j}(\zeta)=I+\frac{\bar{\zeta}_{j}-k_{j}}{\zeta-\bar{\zeta}_{j}} \frac{\left|v_{j}\right\rangle\left\langle \overline{v}_{j}\right|}{\left\langle \overline{v}_{j} | v_{j}\right\rangle},
\]
and its inverse matrix
\[
\Gamma_{j}(\zeta)^{-1}=I+\frac{\zeta_{j}-\bar{\zeta}_{j}}{\zeta-\zeta_{j}} \frac{\left|v_{j}\right\rangle\left\langle \overline{v}_{j}\right|}{\left\langle \overline{v}_{j} | v_{j}\right\rangle},
\]
where $\begin{array}{l}
\left|v_{j}\right\rangle \in \operatorname{Ker}\left(P_{+} \Gamma_{1}^{-1} \cdots \Gamma_{j-1}^{-1}\left(\zeta_{j}\right)\right),\left\langle \overline{v}_{j}|=| v_{j}\right\rangle^{\dagger}.
\end{array}$
Now, introducing the matrix function:
$$
\Gamma(\zeta)=\Gamma_{N}(\zeta)\Gamma_{N-1}(\zeta) \cdots \Gamma_{1}(\zeta),
$$
a calculation gives
\begin{equation}
\begin{aligned}\notag
\Gamma(\zeta) &=I+\sum_{j, k=1}^{N} \frac{\left|v_{j}\right\rangle\left(M^{-1}\right)_{j k}\left\langle \overline{v}_{k}\right|}{\zeta-\bar{\zeta}_{k}}
\end{aligned},
\end{equation}
\begin{equation}
\begin{aligned}\notag
\Gamma^{-1}(\zeta) &=I-\sum_{j, k=1}^{N} \frac{\left|v_{j}\right\rangle\left(M^{-1}\right)_{j k}\left\langle \overline{v}_{k}\right|}{\zeta-\zeta_{j}}
\end{aligned},
\end{equation}
where $M$ is a $N \times N$ matrix with its $(j, k)$ th element given by
\begin{equation}\label{2.74}
\begin{array}{c}
M_{j k}=\frac{\left\langle \overline{v}_{j}|v_{k}\right\rangle}{\bar{\zeta}_{j}-\zeta_{k}}, \quad 1 \leq j, k \leq N. \\
\end{array}
\end{equation}
Based on the above argument, we are confident that  $\Gamma(T, \zeta)$ cancels all the zeros of $P_{\pm},$ and the analytic solutions can be represented as
\[
\begin{array}{l}
P^{+}(\zeta)=\widehat{P}^{+}(\zeta) \Gamma(\zeta), \\
P^{-}(\zeta)=\Gamma^{-1}(\zeta) \widehat{P}^{-}(\zeta).
\end{array}
\]
Here, $\widehat{P}^{\pm}(\zeta)$ are meromorphic $2 \times 2$ matrix functions in $\mathbb{C}_{+}$ and $\mathbb{C}_{-}$, respectively, with finite number of poles and specified residues.  Therefore, all zeros of the RH problem have been eliminated, and a regular RH problem
\[
\widehat{P}^{-}(\zeta) \widehat{P}^{+}(\zeta)=\Gamma(\zeta) G(\zeta) \Gamma^{-1}(\zeta), \quad \zeta \in \mathbb{R},
\]
with boundary condition:  $\widehat{P}^{\pm}(\zeta)=P^{\pm}(\zeta)\Gamma^{-1} \rightarrow I$ as $\zeta \rightarrow \infty$ can be formulated. Then $ P^{+}(\zeta)=\Gamma$  when $\zeta \rightarrow \infty$.

\subsection{Solution of the Riemann-Hilbert Problem}
\ \ \ \
In this subsection, how to solve the matrix RH problem \eqref{2.44} in the complex $\zeta$ plane is discussed. In most of these discussions, $T$ is a dummy variable, hence will be
suppressed in our notations. Thus, if we expand $P$ at large $\zeta$ as
\begin{equation}\label{2.47}
P(T,\zeta) = I +\zeta^{-1}P_{1}^{\pm}(T)+O(\zeta^{-2} ), \quad \zeta\rightarrow\infty,
\end{equation}
and insert \eqref{2.47} into \eqref{2.6} and \eqref{2.30}, then by comparing terms of the same power in $\zeta^{-1}$, we find at $O(1)$ that
\begin{equation}\label{2.48}
Q = i[\wedge,P_{1}^{+}]=-i[\wedge,P_{1}^{-}].
\end{equation}
Hence the solution $u$ can be reconstructed by
\begin{equation}\label{2.49}
u = 2i( P_{1}^{+})_{12}= -2i(P_{1}^{-})_{12}.
\end{equation}
This completes the inverse scattering process. 
Continuing the above calculations, at $O(\zeta^{-1})$ in \eqref{2.6}, we get
\begin{equation}\label{2.50}
diag(P^{+}_{1})_{T} = diag(QP^{+}_{1}).
\end{equation}

\subsection{Time Evolution of Scattering Data}
\ \ \ \
In this subsection, we determine the time evolution of the scattering data. First, the time evolution of the scattering matrices $S$ and $S^{-1}$ is analyzed. The definition \eqref{2.17} for the scattering matrix can be rewritten as
\[
J_{-} E=J_{+} E S, \quad \zeta \in \mathbb{R}.
\]
Since $J_{\pm}$ satisfies the temporal equation \eqref{2.6} of the Lax pair, then multiply \eqref{2.6} by the time-independent diagonal matrix $E=e^{-i \zeta \Lambda T}$. Due to $J_{-} E,$ i.e., $J_{+} E S,$ satisfies the same temporal equation \eqref{2.6} as well. Thus, by inserting $J_{+} E S$ into \eqref{2.6}, taking the limit $T \rightarrow+\infty,$ and recalling the boundary condition \eqref{2.10} for $J_{+}$ as well as the fact that $V \rightarrow 0$
as $T \rightarrow \pm \infty$, then
\[
S_{Z}=-(4i\beta\zeta^3+2i\alpha\zeta^{2})[\Lambda, S].
\]
Similarly, inserting $J_{-} E S^{-1}$ into \eqref{2.6}, taking the limit $T \rightarrow-\infty,$ and recalling the asymptotics \eqref{2.10} for $J_{-},$ 
\begin{equation}\label{2.103}
\left(S^{-1}\right)_{Z}=-2 i \zeta^{2}\left[\Lambda, S^{-1}\right].
\end{equation}
From these two equations, 
\begin{equation}\label{2.104}
\frac{\partial \hat{s}_{22}}{\partial Z}=\frac{\partial s_{22}}{\partial Z}=0,
\end{equation}
and
\begin{equation}\label{2.105}
\frac{\partial \hat{s}_{12}}{\partial Z}=-(4i\beta\zeta^3+2i\alpha\zeta^{2}) \hat{s}_{12}, \quad \frac{\partial s_{21}}{\partial Z}=(4i\beta\zeta^3+2i\alpha\zeta^{2}) s_{21}
\end{equation}
 can be derived.
The equation\eqref{2.104} shows that $\hat{s}_{22}$ and $s_{22}$ are time independent. Recall that $\zeta_{k}$ and $\bar{\zeta}_{k}$ are zeros of $|P^{\pm}(\zeta)|$, i.e., they are zeros of $\hat{s}_{22}(\zeta)$ and $s_{22}(\zeta)$ in view of \eqref{2.52}. Thus $\zeta_{k}$ and $\bar{\zeta}_{k}$ are also time independent. The two equations in \eqref{2.105} give the time evolution for the scattering data $\hat{s}_{12}$ and $s_{21},$ which is
\[
\hat{s}_{12}(Z ; \zeta)=\hat{s}_{12}(0 ; \zeta) e^{-(4i\beta\zeta^3+2i\alpha\zeta^{2})Z}, \quad s_{21}(Z ; \zeta)=s_{21}(0 ; \zeta) e^{(4i\beta\zeta^3+2i\alpha\zeta^{2})Z}.
\]
Next we determine the time dependence of the scattering data $v_{k}$ and $\left\langle \overline{v}_{j}\right|$. This determination is similar to that for the $T$-dependence of $v_{k}$ and $\left\langle \overline{v}_{j}\right|$ at the end of the previous subsection. We also start with \eqref{2.67} for $\left|v_{k}\right\rangle$ and $\left\langle \overline{v}_{j}\right|$. Taking the time derivative to the $\left|v_{k}\right\rangle$ equation and recalling that $P^{+}$ satisfies the temporal equation $\eqref{2.67}$, then
\[
P^{+}\left(\zeta_{k} ; T,Z\right)\left(\frac{\partial \left|v_{k}\right\rangle}{\partial Z}+(4i\beta\zeta^3+2i\alpha\zeta^{2})\Lambda \left|v_{k}\right\rangle\right)=0,
\]
i.e,
\[
\frac{\partial \left|v_{k}\right\rangle}{\partial Z}+(4i\beta\zeta^3+2i\alpha\zeta^{2})\left|v_{k}\right\rangle=0.
\]
Combining it with the spatial dependence \eqref{2.99}, we get the temporal and spatial dependence for the vector $\left|v_{k}\right\rangle$ as
\begin{equation}\label{2.109}
\left|v_{k}\right\rangle(T, Z)=e^{-i \zeta \Lambda T-(4i\beta\zeta^3+2i\alpha\zeta^{2})\Lambda Z} \left|v_{k0}\right\rangle,
\end{equation}
where $\left|v_{k0}\right\rangle$ is a constant. Similar calculations for $\left\langle \overline{v}_{k}\right|$ give
\[
\left\langle \overline{v}_{k}\right|(T, Z)=\left\langle \overline{v}_{k0}\right| e^{i \bar{ \zeta_{k}}T+(4i\beta\overline{\zeta}^3+2i\alpha\overline{\zeta}^{2})Z}.
\]

The scattering data needed to solve this non-regular RH problem is
\begin{equation}\label{2.95}
\left\{s_{21}(\xi), \hat{s}_{12}(\xi), \xi \in \mathbb{R} ; \quad \zeta_{k}, \bar{\zeta}_{k}, \left|v_{j}\right\rangle, \left\langle \overline{v}_{k}\right|, 1 \leq k \leq N\right\}
\end{equation}
 which is called the minimal scattering data. Of this scattering data, vectors $\left|v_{k}\right\rangle$ and $\left\langle \overline{v}_{k}\right|$ are
$T$ dependent, while the others are not. From this scattering data at any later time, we can solve the non-regular RH problem \eqref{2.44} with zeros \eqref{2.67}, and thus reconstruct the solution $u(T, Z)$ at any later time from the formula \eqref{2.49}. This completes the inverse scattering transform process for solving the Hirota equation \eqref{hirota}.

\section{N-Soliton Solutions for the Hirota and IVC-Hirota equation}

In this section, the N-soliton formula for Hirota is derived by considering the simple zeros in the RH problem. Then the  N-soliton formula  for the IVC-Hirota equation  can be constructed by a special transformation.
We also give the dynamic analysis for  some interesting exact solutions of the Hirota and IVC-Hirota equation.

\subsection{N-Soliton matrix for the Hirota and IVC-Hirota equation}

It is well known that when scattering data $\hat{s}_{12}=s_{21}=0$, the soliton solutions  correspond to the reflectionless potential. Then jump matrix $G = I$, $\widehat{G} = 0$. Due to $ P^{+}(\zeta)=\Gamma, \zeta\rightarrow\infty$. Recall to \eqref{2.49}, we can get
\begin{equation}\label{2.120}
u(T,Z)=2i(\sum_{j,k=1}^{N}\left|v_{j}\right\rangle(M^{-1})_{jk}\left\langle \overline{v}_{k}\right|)_{12}.
\end{equation}
Here vectors $v_{j}$ are given by \eqref{2.109}, $\left\langle \overline{v}_{k}\right|=\left|v_{k}\right\rangle^{\dag}$, and matrix $M$ is given by \eqref{2.74}. Without loss of generality, taking $\left|v_{k0}\right\rangle= (c_{k}, 1)^{T}$ in the following discussion. In addition, introduce the notation
\begin{equation}\label{2.121}
\theta_{k}=-i\zeta_{k}T-(4i\beta \zeta_{k}^{3}+2i\alpha \zeta_{k}^{2})Z.
\end{equation}
Then the above solution $u$ can be written out explicitly as
\begin{equation}\label{2.122}
u(T,Z)=2i\sum_{j,k=1}^{N}c_{j}e^{\theta_{j}-\theta_{k}^{*}}(M^{-1})_{jk},
\end{equation}
where the elements of the $N \times N$ matrix $M$ are given by
\begin{equation}\label{2.123}
M_{jk}=\frac{e^{-(\theta_{k}+\theta_{j}^{*})}
+c_{j}^{*}c_{k}e^{\theta_{k}+\theta_{j}^{*}}}{\zeta_{j}^{*}-\zeta_{k}}.
\end{equation}
Notice that $M^{-1}$ can be expressed as the transpose of $M's$ cofactor matrix divided by $|M|$. In addition, remember that the determinant of a matrix can be expressed as the sum of its elements along a row or column multiplying their corresponding cofactor. Therefore, the solution of the general Hirota equation can be rewritten as
\begin{equation}\label{2.124}
u(T,Z) = -2i\frac{|F|}{|M|},
\end{equation}
where $F$ is the following $(N +1)\times(N +1)$ matrix
\begin{equation}
\left(\begin{matrix}
0&e^{-\theta_{1}^{*}}&...&e^{-\theta_{N}^{*}}\\
c_{1}e^{\theta_{1}}&M_{11}&...&M_{N1}\\
.&.&.&.\\
.&.&.&.\\
.&.&.&.\\
c_{N}e^{\theta_{N}}&M_{1N}&...&M_{NN}
\end{matrix}
\right).
\end{equation}

There is a transformation relationship \begin{equation}\label{tr}
q=f(z)u(T, Z) e^{ig(t, z)},
\end{equation}
which can map IVC-Hirota equation \eqref{vcheq} to the Hirota equation \eqref{hirota}.
Where,
\begin{equation}\notag
\begin{gathered}
f(z)=\sqrt{\frac{\alpha_{1}}{\alpha_{4}}},\enspace 
Z=-\frac{\sqrt{2 \delta}}{12 \beta} \int \alpha_{1} d z, \enspace
T=\frac{\sqrt{2 \delta}}{2}( t-(\frac{\alpha^{2}
}{36 \beta^{2}}- \delta)\int \alpha_{1} dz), \\ and \enspace
g(t, z)=-\frac{6 \beta \delta+\alpha \sqrt{2 \delta}}{6 \beta} t-\frac{216 \delta^{2} \beta^{3}+54 \alpha \beta^{2} \delta \sqrt{2 \delta}-\alpha^{3} \sqrt{2 \delta}}{324 \beta^{3}} \int \alpha_{1} dz.
\end{gathered}
\end{equation}
Transformation relationship \eqref{tr} reveals the integrability of IVC-Hirota equation.

Using this transformation relationship, the $n$-soliton matrix solution of IVC-Hirota equation can be constructed as follows
 \begin{equation}\label{vchs}
 q=-2i\frac{|F|}{|M|}f(z)e^{ig(t, z)}.
\end{equation}

\subsection{Exact solutions for the Hirota and IVC-Hirota equation}

When $N =1$, $c_{1}=1$ and $ \zeta_{1}=\xi+i\eta$, then the one-soliton solution for the Hirota equation can be derived from \eqref{2.124} 
 \begin{equation}
u_{1}=2\eta e^{(-8i\beta \xi^{3}-4i\alpha \xi^{2}+24i\beta\eta^{2}\xi+4i\eta^{2})Z+(1-2i\xi)T}sech(2\eta(-4\beta \eta^{2}+12\beta \xi^{2}+4\alpha \xi)Z+2\eta T),
\end{equation}
which is a classical bell-shaped soliton solution.

Using the transformation relationship \eqref{tr}, the one-soliton solution of IVC-Hirota equation is shown as
\begin{equation}
q_{1}=2\sqrt{\frac{\alpha_{1}(z)}{\alpha_{4}(z)}}\eta e^{\frac {A}{324\,{\beta}^{3}}}sech(B),
\end{equation}
\begin{equation}\notag
\begin{split}
A&=(\xi\,\beta+\frac{\alpha}{6})(216\,i\sqrt {2\delta}( {\xi}^{2}{\beta}^{2}-3\,{\eta}^{2} {\beta}^{2}+\frac{\alpha\,\beta\,\xi}{3}+\frac{\,{\alpha}^{2} }{36})-324\,i\sqrt {2}
  {\beta}^{2}{\delta}^{\frac{3}{2}}-216\,i{\beta}^{3}{\delta}^{2} ) \int \!\alpha_{1}(z) \,{\rm d}z\\&-324\,i{\beta}^{2}
( \sqrt{2\delta}( \xi\,\beta+\frac{\alpha}{6})+
\beta\,\delta)t,\\
B&=\eta\sqrt {2\delta} ((\frac{2}{3}\,{\eta
}^{2}-2\,{\xi}^{2}+\delta-\frac{2\alpha\xi}{3\beta}-\frac{{\alpha
}^{2}}{18{\beta}^{2}}) \int \!\alpha1(z) \,{\rm d}z+t).
\end{split}
\end{equation}
The introduction of the integral term enriches the dynamic behavior of the solution of the variable coefficient equation. It can be seen from the expression that $\alpha_{1}(z)$ affects the trajectory of the solution $q_{1}$, and $\frac{\alpha_{1}(z)}{\alpha_{4}(z)}$ affects the amplitude.  The central trajectory equation of solution $q_{1}$ is
$$t=(2\,{\xi}^{2}-\frac{2}{3}\,{\eta
}^{2}-\delta+\frac{2\alpha\xi}{3\beta}+\frac{{\alpha
}^{2}}{18{\beta}^{2}}) \int \!\alpha1(z) \,{\rm d}z.$$
When $2\,{\xi}^{2}-\frac{2}{3}\,{\eta
}^{2}-\delta+\frac{2\alpha\xi}{3\beta}+\frac{{\alpha
}^{2}}{18{\beta}^{2}}=0$, the dynamic image $ q_{1}$ shows the shape of a common bell soliton which is similar to the soliton solution of constant coefficient equation. 

To simplify the expression of the solution $q_{1}$, taking  $\alpha=0$, $\xi=1$ and  $\eta=1$,  then
 \begin{equation}
|q_{1}|^{2}=4\,\left|{\frac{\alpha_{1}(z)}{\alpha_{4}(z)}}\right|{\rm sech}^{2}(\sqrt {2\delta}((\delta-\frac{4}{3})
\int \!\alpha_{1}(z){\rm d}z+t)).
\end{equation}
It is needed to note that when $\alpha=0$, the value of $\beta$ has no effect on the solution $q_{1}$. The central trajectory equation of solution $q_{1}$ is
$t=(\frac{4}{3}-\delta) \int \!\alpha1(z) \,{\rm d}z.$  The dynamic image $ q_{1}$ shows the shape of a bell soliton when $\delta=\frac{4}{3}$. In other words, due to the integral term, the form of the solution becomes more abundant. 
The dynamic evolution diagram of $|q_{1}|^{2}$ is symmetric about  $z$ axis (i.e. $|q_{1}|^{2}$ is an even function about $z$) when  $\alpha_{1}$ and $\alpha_{4}$ are odd numbers. Moreover, the amplitude of the solution $q_{1}$ is constant $2$ with $\alpha_{1}=\pm \alpha_{4}$.  $\alpha_{1}(z)$ and $\alpha_{4}(z)$ are dispersion and nonlinear effects, respectively. Dispersion broadens the waveform and nonlinear effects narrow it. Under certain conditions, the two effects reach a balance and maintain waveform stability. Next, we study the specific effects of different nonlinear terms and dispersion terms on the dynamic behaviour of the solutions. In order to further study the effect of the dispersion and nonlinear term on the dynamics of the solution, we give the different excitation states of  $\alpha_{1}$ and $\alpha_{4}$.  
 
First, we let the dispersion term be in the simplest polynomial form. For example, fixed the coefficient $\alpha_{1}(z)=z^{n}$, then the center trajectory equation is $t=\frac{1}{n+1}(\frac{4}{3}-\delta)z^{n+1}$.  It can  also be seen that the value of $\delta$ has a great influence on the propagation path of the solution. 
In particular, when $\delta=\frac{4}{3}$, the dynamic behavior of the solution of the variable coefficient equation is similar to that of the constant coefficient equation. 
When $n=1$, we can get $|q_{1}|^{2}=4\left|\frac{z}{\alpha_{4}}\right|sech^{2}(\sqrt {2\delta}((\delta-\frac{4}{3})\frac{1}{2}z^{2}+t))$ and its center trajectory equation is $t=\frac{1}{6}z^{2}$. In order to construct the meaningful solutions which are non-singularity and convergent, we can take $\alpha_{4}=z$, then the amplitude of  $q_{1}$ is $2$. The dynamic behaviour of the one-soliton solution in this case takes on the form of a parabola symmetric on the $z$ axis, which can be seen in the Fig.\eqref{q111}. We also can take $\alpha_{4}(z)=z^{2}+1$, at this case, the amplitude of the solution increases on the interval $z$ belongs to $(-\infty, -1)$, $(0, 1)$, and decreases on the interval  $z$ belongs to $(-1, 0)$, $(1, +\infty)$, with the minimum value $0$ at $z=0$ and the maximum value $\sqrt{2}$ at $z=\pm1$, which can be seen in the Fig.\eqref{q112}. When the dispersion term $\alpha_{1}(z)$ takes the form of other polynomials, such as $z+1$, we can see the dynamic evolution diagram  in  Fig.\eqref{q113} which is similar to  $\alpha_{1}(z)=z$ except axis of symmetry. When $n=2$, then the  center trajectory equation is $t=\frac{1}{3}(\frac{4}{3}-\delta)z^{3}$. Taking  $\alpha_{4}(z)=z^{2}$, the dynamic evolution diagram is plotted in Fig.\eqref{q121}. Taking $\alpha_{4}=z^{2}+1$, the dynamic evolution diagram of the solution is presented in Fig.\eqref{q122}. Let $\alpha_{1}(z)=1+z^{2}$ and $\alpha_{4}=1+10z^{2}$, we can obtain a solution with the amplitude maximizes at the origin and decreases as $z$ goes to infinity, which can be seen in Fig.\eqref{q123}. 

\begin{figure}[ht!]
\centering
\subfigure[]{\label{q111}
\begin{minipage}[b]{0.25\textwidth}
\includegraphics[width=3.4cm]{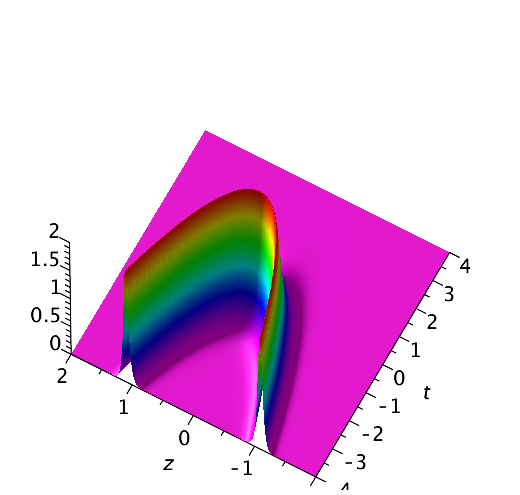}
\end{minipage}
}
\subfigure[]{\label{q112}
\begin{minipage}[b]{0.25\textwidth}
\includegraphics[width=3.4cm]{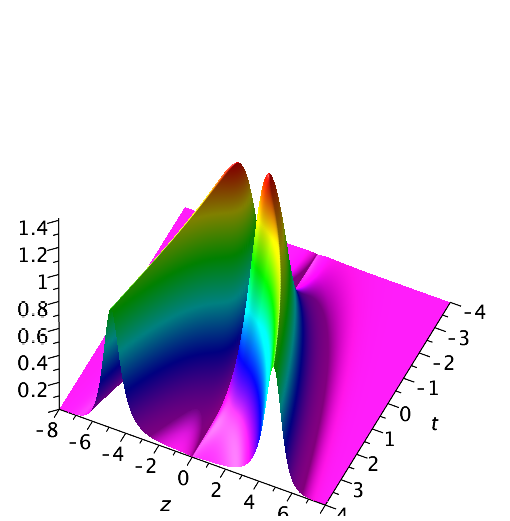}
\end{minipage}
}
\subfigure[]{\label{q113}
\begin{minipage}[b]{0.25\textwidth}
\includegraphics[width=3.4cm]{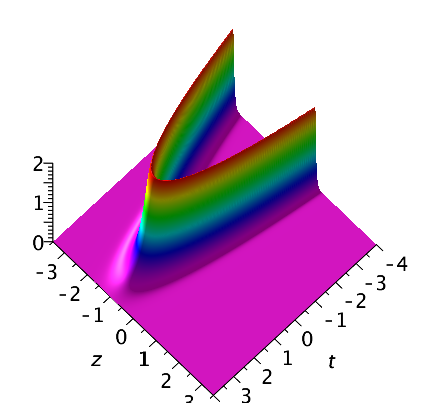}
\end{minipage}
}
\subfigure[]{\label{q121}
\begin{minipage}[b]{0.25\textwidth}
\includegraphics[width=3.4cm]{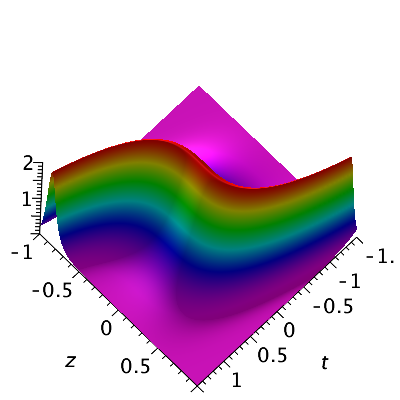}
\end{minipage}
}
\subfigure[]{\label{q122}
\begin{minipage}[b]{0.25\textwidth}
\includegraphics[width=3.4cm]{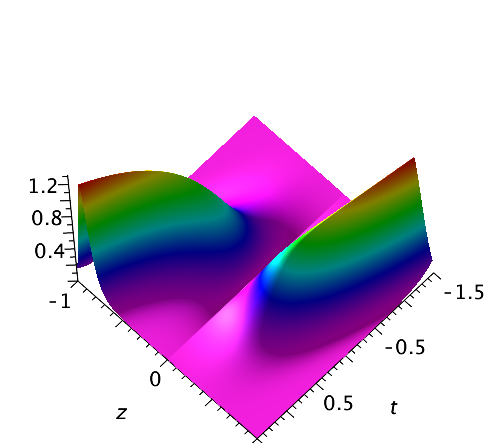}
\end{minipage}
}
\subfigure[]{\label{q123}
\begin{minipage}[b]{0.25\textwidth}
\includegraphics[width=3.4cm]{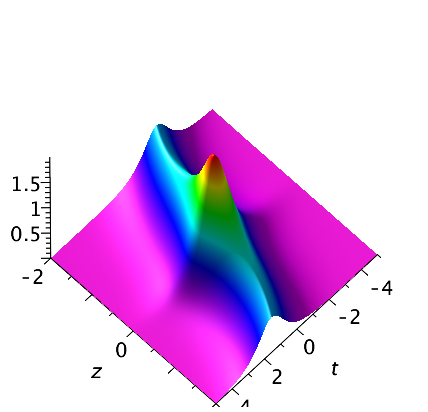}
\end{minipage}
}

\caption{ The evolution plot of 1-soliton solutions in the IVC-Hirota equation: (a) $\alpha_{1}(z)=\alpha_{4}(z)=z$ and $\delta=8$;
(b) $\alpha_{1}(z)=z$, $\alpha_{4}(z)=z^{2}+1$ and $\delta=1$;
(c) $\alpha_{1}(z)=z+1$, $\alpha_{4}(z)=1+z$ and $\delta=5$;
(d) $\alpha_{1}(z)=z^{2}$, $\alpha_{4}(z)=z^{2}$ and $\delta=8$;
(e) $\alpha_{1}(z)=z^{2}$, $\alpha_{4}(z)=z^{2}+1$ and $\delta=8$;
(f) $\alpha_{1}(z)=z^{2}+1$, $\alpha_{4}(z)=10z^{2}+1$ and $\delta =1$.}
\label{1z}
\end{figure}

Second, we can obtain $|q_{1}|^{2}=4\left|\frac{sin(kz)}{\alpha_{4}(z)}\right|sech^{2}(\sqrt {2\delta}(\frac{1}{k}(\frac{4}{3}-\delta)cos(kz)+t))$ when considering periodic functions  $\alpha_{1}(z)=sin(kz)$ as excitation function. Now the center trajectory of the solution $q_{1}$ is a cosine wave, where $k$ determines the period and $\delta$ has a big effect on the shape of the trajectory. The dynamic evolution diagram of different parameters is plotted in Fig.\eqref{2z}.

\begin{figure}[ht!]
\centering
\subfigure[]{\label{q133}
\begin{minipage}[b]{0.25\textwidth}
\includegraphics[width=4cm]{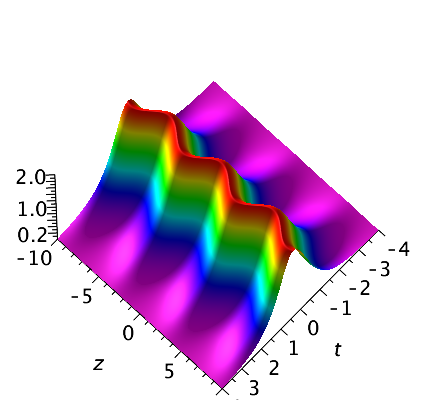}
\end{minipage}
}
\subfigure[]{\label{q1411}
\begin{minipage}[b]{0.25\textwidth}
\includegraphics[width=4cm]{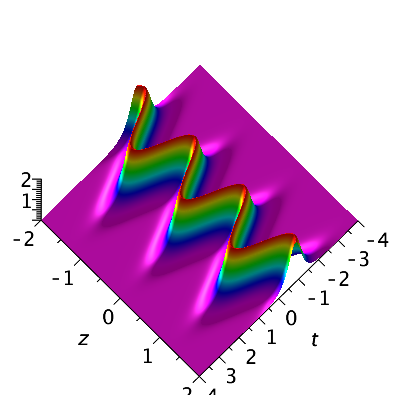}
\end{minipage}
}
\subfigure[]{\label{q135}
\begin{minipage}[b]{0.25\textwidth}
\includegraphics[width=4cm]{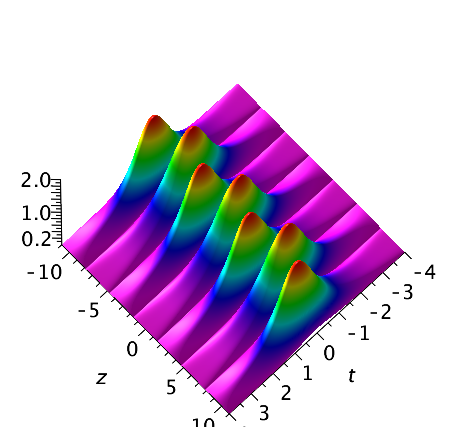}
\end{minipage}
}
\subfigure[]{\label{q1413}
\begin{minipage}[b]{0.25\textwidth}
\includegraphics[width=4cm]{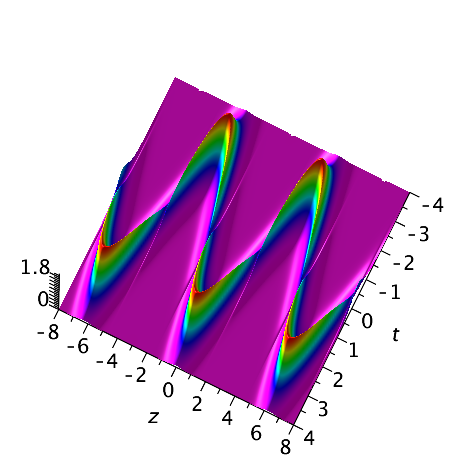}
\end{minipage}
}
\subfigure[]{\label{q1415}
\begin{minipage}[b]{0.25\textwidth}
\includegraphics[width=4cm]{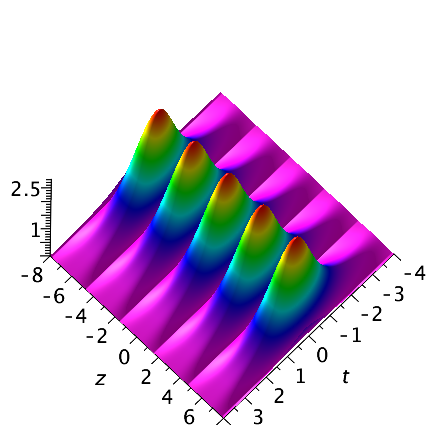}
\end{minipage}
}
\subfigure[]{\label{q1414}
\begin{minipage}[b]{0.25\textwidth}
\includegraphics[width=4cm]{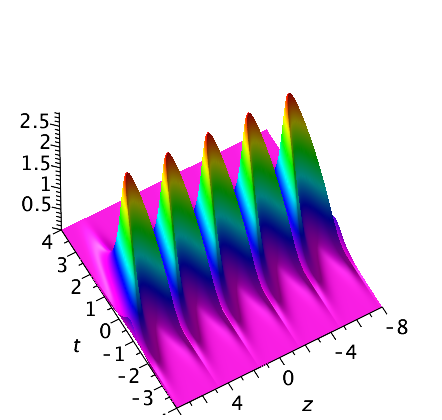}
\end{minipage}
}
\caption{  The evolution plot of 1-soliton solutions in the IVC-Hirota equation: (a) $\alpha_{1}(z)=\alpha_{4}(z)=sin(z)$ and $\delta=1$; 
(b) $\alpha_{1}(z)=\alpha_{4}(z)=sin(5z)$ and $\delta=6$; 
$\alpha_{1}(z)=sin(z)$ and $\alpha_{4}(z)=tan(z)$: ((c) $\delta = 1$; (d) $\delta = 4$;)
$\alpha_{1}(z)=sin(2z)$ and $\alpha_{4}(z)=tan(z)$: ((e)$\delta=1$; (f) $\delta = 3$.)}
\label{2z}
\end{figure}

Besides, fixed the coefficient $\alpha_{1}(z)=tanh(z)$, then we can get $|q_{1}|^{2}=4\left|\frac{tanh(z)}{\alpha_{4}}\right|sech^{2}(\frac{\sqrt{2}}{3}ln|cosh(z)|-\sqrt{2}t).$  Fig.\eqref{3z} shows the dynamic evolution process of the nonlinear term $\alpha_{4}$ and $\delta$ with different values.
\begin{figure}[ht!]
\centering
\subfigure[]{\label{q151}
\begin{minipage}[b]{0.25\textwidth}
\includegraphics[width=4.cm]{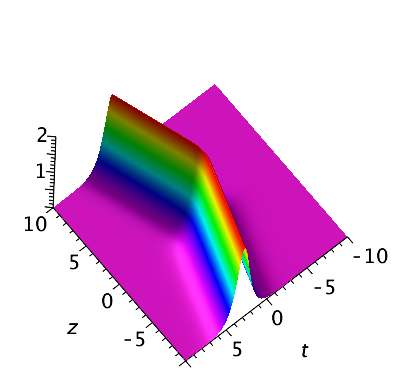}
\end{minipage}
}
\subfigure[]{\label{q15121}
\begin{minipage}[b]{0.25\textwidth}
\includegraphics[width=4.cm]{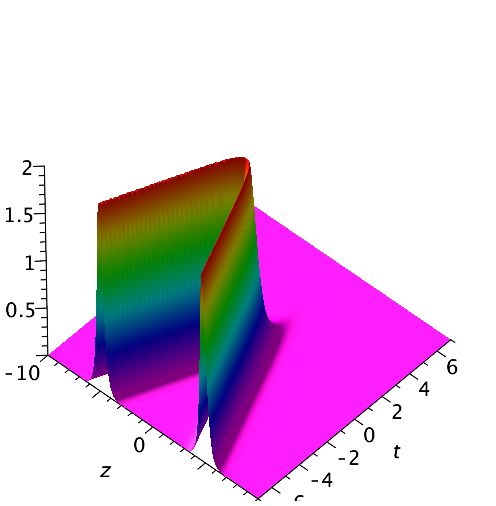}
\end{minipage}
}
\subfigure[]{\label{q1521}
\begin{minipage}[b]{0.25\textwidth}
\includegraphics[width=4.cm]{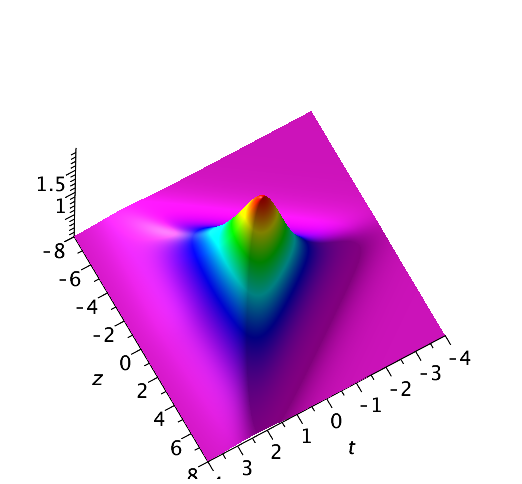}
\end{minipage}
}
\subfigure[]{\label{q1522}
\begin{minipage}[b]{0.25\textwidth}
\includegraphics[width=4.cm]{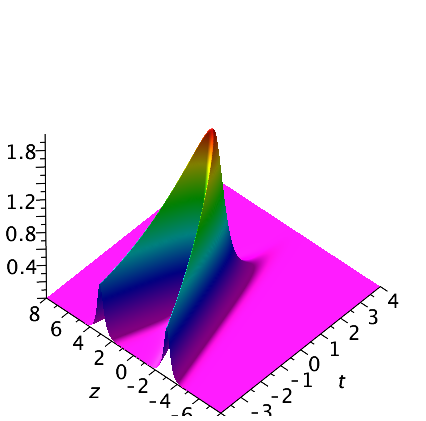}
\end{minipage}
}
\subfigure[]{\label{q153}
\begin{minipage}[b]{0.25\textwidth}
\includegraphics[width=4cm]{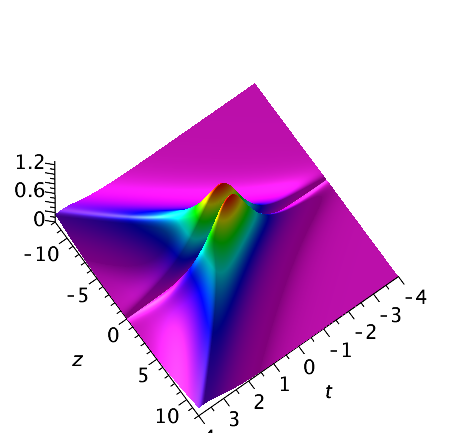}
\end{minipage}
}
\subfigure[]{\label{q1531}
\begin{minipage}[b]{0.25\textwidth}
\includegraphics[width=4cm]{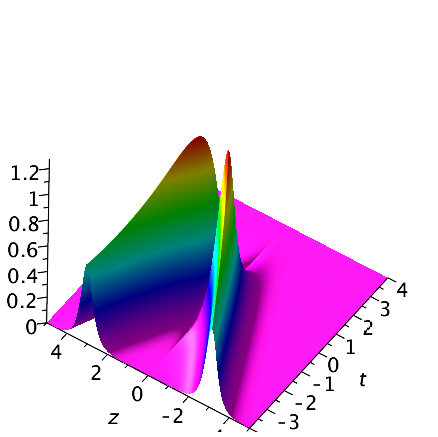}
\end{minipage}
}
\caption{ The evolution plot of 1-soliton solutions in the IVC-Hirota equation: (a) $\alpha_{4}(z)=tanh(z)$ and $\delta=1$; (b) $\alpha_{4}(z)=tanh(z)$ and $\delta=3$; (c) $\alpha_{4}(z)=sinh(z)$ and $\delta=1$; (d) $\alpha_{4}(z)=sinh(z)$ and $\delta=3$; (e) $\alpha_{4}(z)=z^{2}+1$ and $\delta=1$; (f) $\alpha_{4}(z)=z^{2}+1$ and $\delta=3$.}
\label{3z}
\end{figure}
 
When $N=2$, then the two-soliton solution $u_{2}(T,Z)$ of Hirota equation is expressed as follows:
\begin{equation}\label{2.1241}\footnotesize%
\frac{h_{1}\,{{\rm e}^{-\Theta_{1}+\Theta_{4}}}+{h_{2}}\,{{\rm e}^{\Theta_{3}-\Theta_{2}}}-{h_{3}}\,{{\rm e}^{\Theta_{1}+\Theta_{4}}}+{h_{4}}\,{{\rm e}^{\Theta_{3}+\Theta_{2}}}
}{{d_{1}}\,{{\rm e}^{-\Theta_{1}-\Theta_{2}}}+{d_{2}}\,{{\rm e}^{-\Theta_{1}+\Theta_{2}}}+{d_{3}}\,{{\rm e}^{\Theta_{1}+\Theta_{2}}}+{d_{4}}\,{{\rm e}^{\Theta_{3}-\Theta_{4}}}+{ d_{5}}\,{{\rm e}^{
\Theta_{1}-\Theta_{2}}}+{d_{6}}\,{
{\rm e}^{-\Theta_{3}+\Theta_{4}}}},
\end{equation}
where
$$\Theta_{1}=\theta_{1}+\theta_{1}^{*},$$
$$\Theta_{2}=\theta_{2}+\theta_{2}^{*},$$
$$\Theta_{3}=\theta_{1}-\theta_{1}^{*},$$
$$\Theta_{4}=\theta_{2}-\theta_{2}^{*},$$
$$d_{1}=(\zeta^{*}_{1}-\zeta^{*}_{2})({\zeta_{2}}-{\zeta_{1}}),$$$$d_{2}=|{c_{2}}|^{2}(\zeta^{*}_{1}-\zeta_{2})(\zeta^{*}_{2}-\zeta_{1}),$$
$$d_{3}=|{{c_{1}}}^{2}{{c_{2}}}^{2}|(\zeta^{*}_{2}-\zeta^{*}_{1})({\zeta_{1}}-{\zeta_{2}}),$$
$$d_{4}={c_{1}}{c_{2}}^{*}(\zeta^{*}_{1}-\zeta_{1})({\zeta_{2}}-{\zeta_{2}}^{*}),$$
$$d_{5}=|{c_{1}}|^{2}(\zeta^{*}_{2}-\zeta_{1})(\zeta_{1}^{*}-{\zeta_{2}}),$$
$$d_{6}=c_{1}^{*}{c_{2}}(\zeta^{*}_{2}-\zeta_{2})({\zeta_{1}}-\zeta_{1}^{*}),$$
$$h_{1}=-{c_{2}}\, ({\zeta_{2}}-\zeta^{*}_{1})(\zeta^{*}_{2}-\zeta^{*}_{1})(\zeta^{*}_{2}-{\zeta_{2}}),$$
$$h_{2}={c_{1}}\,({\zeta_{1}}-\zeta^{*}_{1})(\zeta^{*}_{2}-\zeta^{*}_{1})(\zeta^{*}_{2}-{\zeta_{1}}),$$
$$h_{3}=-c_{2}|{c_{1}}|^{2}(\zeta^{*}_{2}-{\zeta_{2}})(\zeta_{1}-\zeta_{2})({\zeta_{2}^{*}}-{\zeta_{1}}),$$
$$h_{4}={c_{1}}\,|{c_{2}}|^{2}({\zeta_{1}}-{\zeta_{2}})({\zeta_{2}}-\zeta^{*}_{1})({\zeta_{1}}-\zeta^{*}_{1}).
$$
We analyze the asymptotic states of the solution \eqref{2.1241} as $Z \rightarrow \pm \infty$ and  $(\alpha,\beta)$ is non-negative.  Without loss of generality, let $\zeta_{k}=\xi_{k}+i\eta_{k}$ and $|\xi_{1}|>|\xi_{2}|$, this means that at $Z=-\infty,$ soliton-$1$ is on the right side of soliton-$2$ and moves slower. Note also that $\eta_{k}>0$ and $\eta_{2}>\eta_{1}$, since $\zeta_{k} \in \mathbb{C}_{+}.$ In the moving frame with velocity $4\beta\eta_{1}^{2}-12\beta\xi_{1}^{2}-4\alpha\xi_{1}$, $\operatorname{Re}\left(\theta_{1}\right)=\eta_{1}(T-4\beta\eta_{1}^{2}Z+12\beta \xi_{1}^{2}Z+4\alpha \xi_{1}Z)=O(1)$.
It is a consequence of 
\begin{equation}\notag
\operatorname{Re}\left(\theta_{2}\right)=\eta_{2}(T-(4\beta\eta_{1}^{2}-12\beta \xi_{1}^{2}-4\alpha\xi_{1})Z)+4\eta_{2}(\beta(\eta_{1}^{2}-\eta_{2}^{2})+\alpha(\xi_{2}-\xi_{1})+3\beta(\xi_{2}^{2}-\xi_{1}^{2}))Z 
\end{equation}
that
\begin{equation}
u_{2}(T, Z) \rightarrow
\begin{cases}
2 i\left(\zeta_{1}^{*}-\zeta_{1}\right) \frac{c_{1}^{-} e^{\theta_{1}-\theta_{1}^{*}}}{e^{-\left(\theta_{1}+\theta_{1}^{*}\right)}+\left|c_{1}^{-}\right|^{2} e^{\theta_{1}+\theta_{1}^{*}}}, \quad Z \rightarrow-\infty,\\
\\
2 i\left(\zeta_{1}^{*}-\zeta_{1}\right) \frac{c_{1}^{+} e^{\theta_{1}-\theta_{1}^{*}}}{e^{-\left(\theta_{1}+\theta_{1}^{*}\right)}+\left|c_{1}^{+}\right|^{2} e^{\theta_{1}+\theta_{1}^{*}}}, \quad Z \rightarrow+\infty,
\end{cases}
\end{equation}
where $c_{1}^{-}= \frac{c_{1}\left(\zeta_{1}-\zeta_{2}\right)}{\left(\zeta_{1}-\zeta_{2}^{*}\right)}$, $c_{1}^{+}=\frac{c_{1}\left(\zeta_{1}-\zeta_{2}^{*}\right)}{\left(\zeta_{1}-\zeta_{2}\right)}$ and $u_{2}(T, Z) \rightarrow
\begin{cases}
+\infty, \quad Z \rightarrow-\infty,\\
\\
-\infty, \quad Z \rightarrow+\infty.
\end{cases}$ Comparing this expression with \eqref{2.1241}, we see that this asymptotic solution is a single-soliton solution with peak amplitude $2\eta_{1}$ and velocity $4\beta\eta_{1}^{2}-12\beta\xi_{1}^{2}-4\alpha\xi_{1}$. Thus, this soliton does not change its shape and velocity after collision. Its phase has shifted and the phase difference for $u_{2}$ at its limits is $\arg \left(u_{2}\left(Z \sim-\infty\right)\right)-\arg \left(u_{2}\left(Z \sim+\infty\right)\right)$. It is apparent from the above analysis that the values of $(\alpha,\beta)$ influence the velocity, phase of the soliton. 
 
Letting  $\zeta_{1}= 0.1+0.7i$ and $\zeta_{2}=-0.1+0.4i$, when $(\alpha, \beta)$ is set as $(0,1)$, $(1,1)$ and $(1,0)$ respectively, their corresponding dynamic evolution diagrams can be drawn in Figs.\eqref{u21}, \eqref{u22} and \eqref{u23}. In particular, when $\alpha=0$ and $4\beta \eta_{1}^{2}-12\beta\xi_{1}^{2}$=$4\beta \eta_{2}^{2}-12\beta\xi_{2}^{2}$, resonance solitons can be obtained. Taking $\zeta_{1}=1+\sqrt{3}i$ and $\zeta_{2} =2+2\sqrt{3}i$, the resonance solitons solution is shown in  Fig.\eqref{gz2gz}.
\begin{figure}[ht!]
\centering
\subfigure[]{\label{u21}
\begin{minipage}[b]{0.2\textwidth}
\includegraphics[width=4cm]{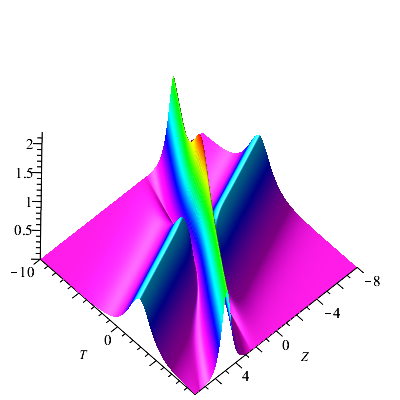} \\
\end{minipage}
}
\subfigure[]{\label{u22}
\begin{minipage}[b]{0.2\textwidth}
\includegraphics[width=4cm]{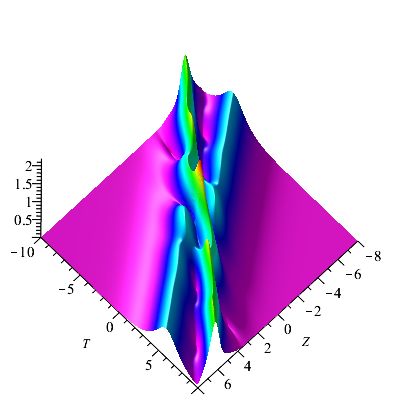} \\
\end{minipage}
}
\subfigure[]{\label{u23}
\begin{minipage}[b]{0.2\textwidth}
\includegraphics[width=4cm]{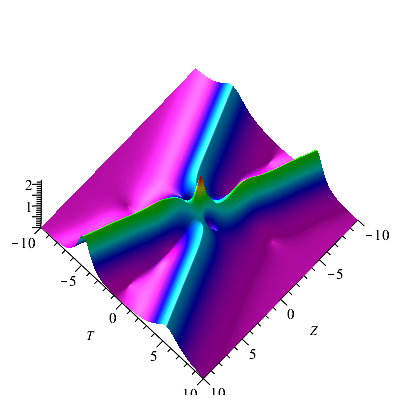}\\
\end{minipage}
}
\subfigure[]{\label{gz2gz}
\begin{minipage}[b]{0.2\textwidth}
\includegraphics[width=4.0cm]{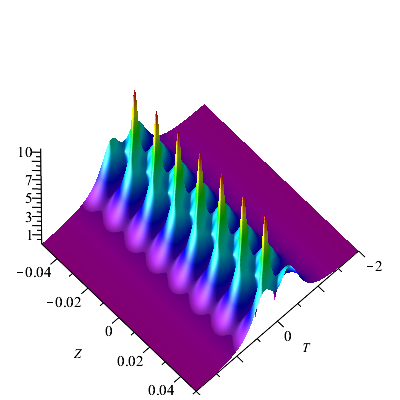}\\
\end{minipage}
}
\caption{ The evolution plot of 2-soliton solutions in the Hirota equation: $\zeta_{1}= 0.1+0.7i$ and $\zeta_{2}=-0.1+0.4i$: ((a) $\alpha=0$ and $\beta=1$; (b) $\alpha=1$ and $\beta=1$; (c) $\alpha=1$ and $\beta=0$.) (d) $\zeta_{1}=1+\sqrt{3}i$, $\zeta_{2} =2+2\sqrt{3}i$, $\alpha=0$ and $\beta=1$.}
\label{11cot(z)}
\end{figure}  

A simpler exact expression of the 2-soliton solution for the IVC-Hirota equation can be acquired by taking $\alpha=0$, $c_{1}=c_{2}=\delta=1$, $\zeta_{1}= i$ and $\zeta_{2}=2i$. Using transformation relationship \eqref{tr},  the expression $q_{2}(t,z)$ is shown as follows: 
\begin{equation}
\label{q2}
q_{2}(t,z)=A_{11}\frac{(e^{\frac{\sqrt{2}}{3}(3t-31f)}-2e^{\frac{2\sqrt{2}}{3}(3t-7f)}-2e^{\frac{4\sqrt{2}}{3}(3t-f)}+e^{\frac{\sqrt{2}}{3}(13f+15t)})}
{16e^{3\sqrt{2}(t-f)}-e^{-12\sqrt{2}f}-9e^{\frac{2\sqrt{2}}{3}(3t-13f)}-9e^{\frac{4\sqrt{2}}{3}(2f+3t)}-e^{6\sqrt{2}(f+t)}},
\end{equation}
where$
f=\int\alpha_{1}(z)dz and
A_{11}=12\sqrt{\frac{\alpha_{1}(z)}{\alpha_{4}(z)}}e^{-\frac{1}{3}i(2f+3t)}.$
From expression \eqref{q2}, we can clearly find that the value of $\beta$ has no effect on the solution $q_{2}$ when $\alpha=0$. Except for the case where both spectral parameters are purely imaginary, let us consider the more general case when $\zeta_{1}= 1+i$, $\zeta_{2}=2+2i$ in the following.

Similar to the case of the one-soliton solution $q_{1}$, the dynamic evolution diagram of $q_{2}$ is symmetric about  $z$ axis when  $\alpha_{1}$ and $\alpha_{4}$ are odd numbers. It can  also be seen that the value of $\delta$ has a great influence on the propagation path of the solution. Using the relationship $2\,{\xi}^{2}-\frac{2}{3}\,{\eta
}^{2}-\delta+\frac{2\alpha\xi}{3\beta}+\frac{{\alpha
}^{2}}{18{\beta}^{2}}=0$, we can find that the dynamic behaviour of the soliton-1 of the solution $q_{2}$ is similar to the corresponding soliton-1 solution of constant coefficient equation when $\delta=\frac{4}{3}$. For the soliton-2, the dynamic behaviour will be similar to the solution of constant coefficient equation when  $\delta=\frac{16}{3}$.  

When we take $\alpha_{1}(z)=\alpha_{4}(z)=z$, the $2$-soliton solution is a constant amplitude solution and the amplitude of Soliton-$1$ is $2$ and the amplitude of Soliton-$2$ is $4$ in the $2$-soliton. At the intersection of solitons-$1$ and solitons-$2$, the amplitude is superimposed linear. The value of $\delta$ will affect the velocity and direction of the $2$-soliton solution. For soliton-$1$ in the $2$-soliton, when $\delta<\frac{4}{3}$, the soliton-$1$ evolves in the region of $t\geq 0$; when $\delta=\frac{4}{3}$, the shape of soliton-$1$ is similar to the bell shape soliton and propagates along $t=0$; otherwise, the soliton-1 evolves in the region $t\leq 0$. For soliton-2 in the 2-soliton, when $\delta<\frac{16}{3}$, the soliton-2 evolves in the region $t\geq 0$; when $\delta=\frac{16}{3}$, soliton-2 propagates along $t=0$; or else, the soliton-2 evolves in the region $t\leq 0$. Figs.\eqref{q211},\eqref{q212},\eqref{q213},\eqref{q214} and \eqref{q215} illustrate the dynamic behaviour of the 2-solitons solution when $\delta$ taking the value of $1$, $\frac{4}{3}$, $2$, $\frac{16}{3}$ and $6$ respectively. When $\alpha_{1}(z)=z$, $\alpha_{4}(z) = z^{2}+1$, the 3-D plots for the 2-soliton solutions are shown in Figs.\eqref{q221},\eqref{q222},\eqref{q223},\eqref{q224} and \eqref{q225}. As can be seen in Figs.\eqref{q231}, \eqref{q232}, \eqref{q233}, \eqref{q234} and \eqref{q235}, a strong interaction occurs when the soliton collides when taking $\alpha_{1}=\alpha_{4}=z^{2}$ in solution $q_{2}$. At the intersection of two $1$-soliton solutions, a linear superposition of amplitudes appears. The $2$-soliton solutions with $\delta = \dfrac{4}{3}$ and $\delta = \dfrac{16}{3}$ are similar to the $2$-soliton solutions of constant coefficients equation. 

\begin{figure}[ht!]
\centering
\subfigure[]{\label{q211}
\begin{minipage}[b]{0.18\textwidth}
\includegraphics[width=3.2cm]{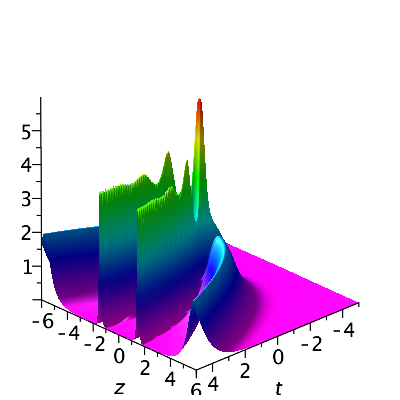}
\end{minipage}
}
\subfigure[]{\label{q212}
\begin{minipage}[b]{0.18\textwidth}
\includegraphics[width=3.2cm]{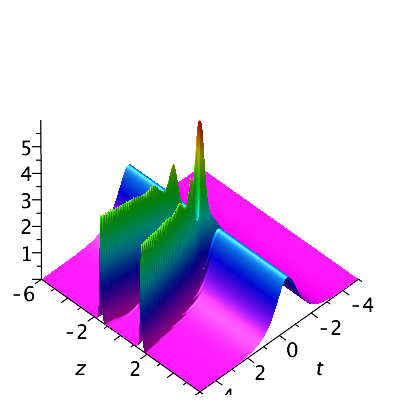}
\end{minipage}
}
\subfigure[]{\label{q213}
\begin{minipage}[b]{0.18\textwidth}
\includegraphics[width=3.2cm]{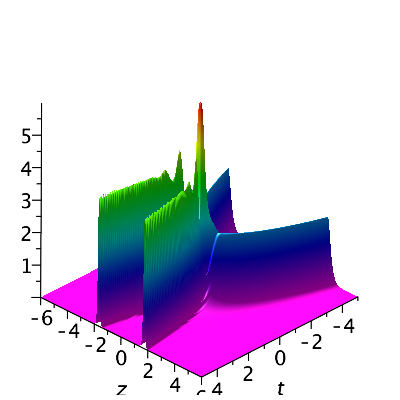}
\end{minipage}
}
\subfigure[]{\label{q214}
\begin{minipage}[b]{0.18\textwidth}
\includegraphics[width=3.2cm]{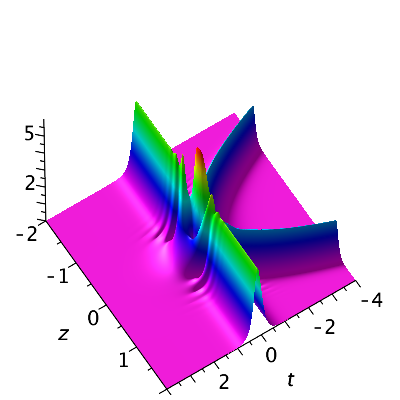}
\end{minipage}
}
\subfigure[]{\label{q215}
\begin{minipage}[b]{0.18\textwidth}
\includegraphics[width=3.2cm]{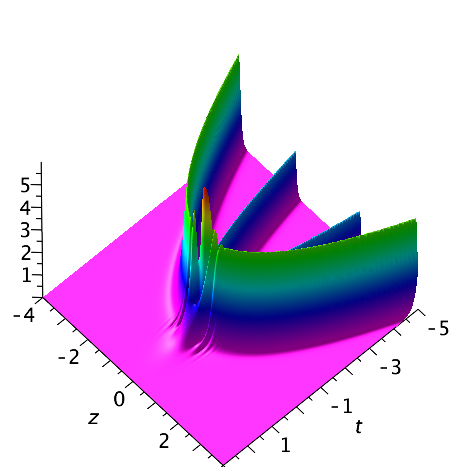}
\end{minipage}
}
\subfigure[]{\label{q221}
\begin{minipage}[b]{0.18\textwidth}
\includegraphics[width=3.2cm]{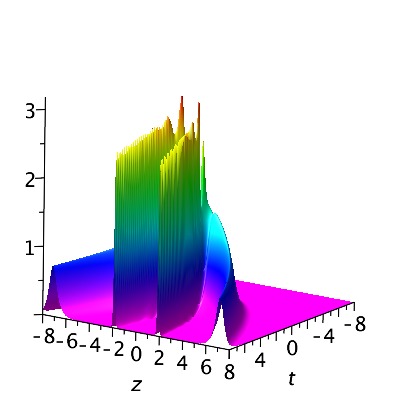}
\end{minipage}
}
\subfigure[]{\label{q222}
\begin{minipage}[b]{0.15\textwidth}
\includegraphics[width=3cm]{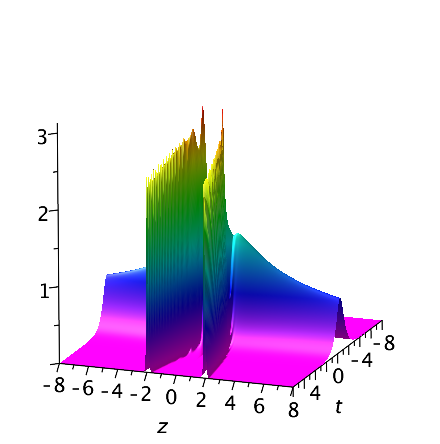}
\end{minipage}
}
\subfigure[]{\label{q223}
\begin{minipage}[b]{0.18\textwidth}
\includegraphics[width=3.2cm]{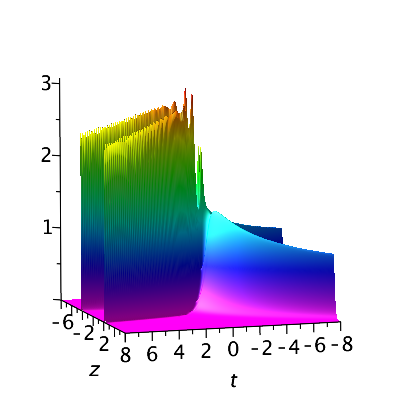}
\end{minipage}
}
\subfigure[]{\label{q224}
\begin{minipage}[b]{0.18\textwidth}
\includegraphics[width=3.2cm]{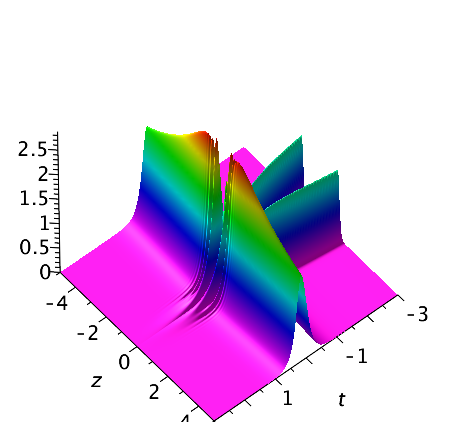}
\end{minipage}
}
\subfigure[]{\label{q225}
\begin{minipage}[b]{0.18\textwidth}
\includegraphics[width=3.2cm]{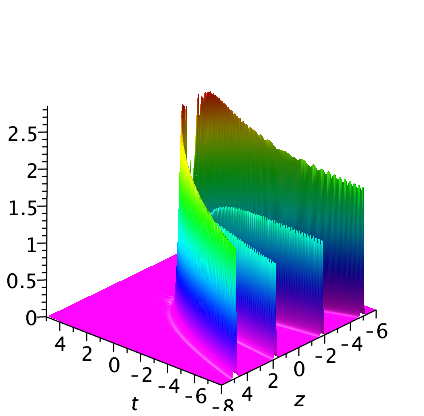}
\end{minipage}
}
\subfigure[]{\label{q231}
\begin{minipage}[b]{0.18\textwidth}
\includegraphics[width=3.2cm]{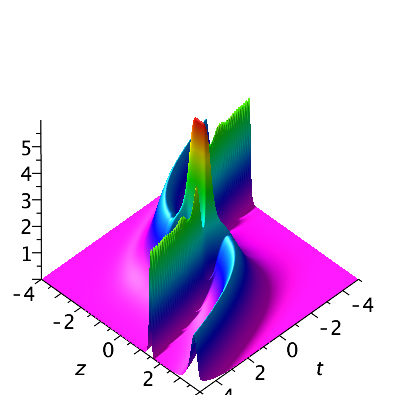}
\end{minipage}
}
\subfigure[]{\label{q232}
\begin{minipage}[b]{0.18\textwidth}
\includegraphics[width=3.2cm]{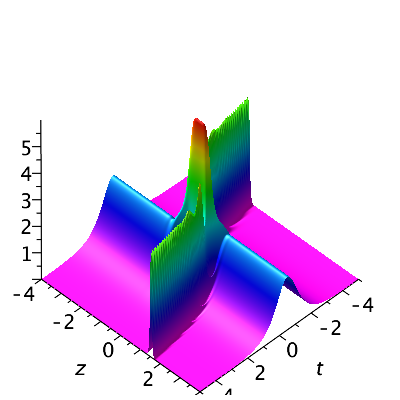}
\end{minipage}
}
\subfigure[]{\label{q233}
\begin{minipage}[b]{0.18\textwidth}
\includegraphics[width=3.2cm]{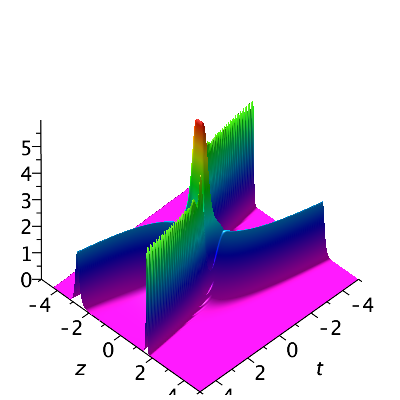}
\end{minipage}
}
\subfigure[]{\label{q234}
\begin{minipage}[b]{0.18\textwidth}
\includegraphics[width=3.2cm]{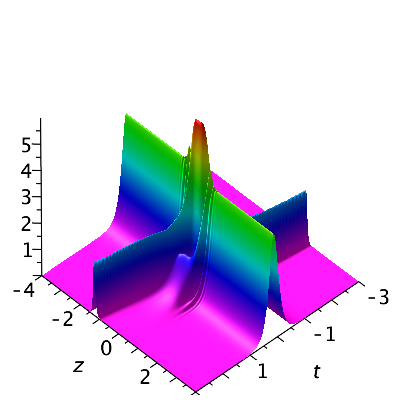}
\end{minipage}
}
\subfigure[]{\label{q235}
\begin{minipage}[b]{0.18\textwidth}
\includegraphics[width=3.2cm]{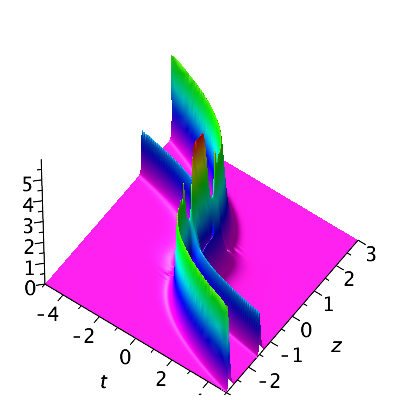}
\end{minipage}
}
\caption{ The evolution plot of 2-soliton solutions in the IVC-Hirota equation: $\alpha_{1}(z)=\alpha_{4}(z)=z$: ((a) $\delta = 1$; (b) $\delta = \frac{4}{3}$; (c) $\delta = 2$; (d) $\delta = \frac{16}{3}$; (e) $\delta = 6$.) 
$\alpha_{1}(z)=z$ and $\alpha_{4}(z) = z^{2}+1$: ((f) $\delta = 1$; (g) $\delta = \frac{4}{3}$; (h) $\delta = 2$; (i) $\delta = \frac{16}{3}$; (j) $\delta = 6$.)
$\alpha_{1}=\alpha_{4}=z^{2}$: ((k) $\delta = 1$; (l) $\delta = \frac{4}{3}$; (m) $\delta = 2$; (n) $\delta = \frac{16}{3}$; (o) $\delta = 6$.)}
\label{11cot(z)}
\end{figure}  

Taking  $\alpha_{1}(z)=sin(kz)$ as excitation function, the parameter values of $k$ and $\delta$ have great influence on the shape of the soliton solutions. The dynamic evolution diagram of different parameters can be seen in 
Fig.\eqref{q2sin}.
\begin{figure}[ht!]
\centering
\subfigure[]{\label{q261}
\begin{minipage}[b]{0.3\textwidth}
\includegraphics[width=4cm]{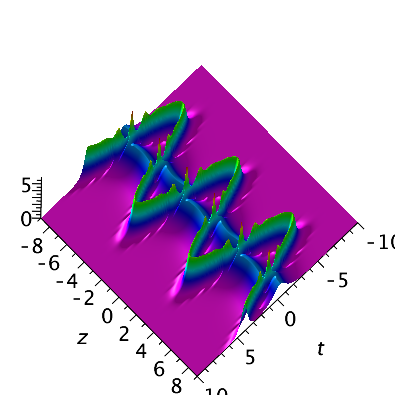} 
\end{minipage}
}
\subfigure[]{\label{q262}
\begin{minipage}[b]{0.3\textwidth}
\includegraphics[width=4cm]{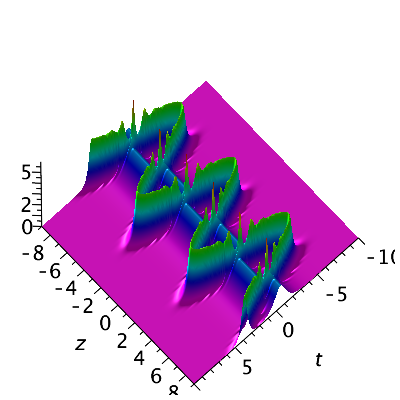}
\end{minipage}
}
\subfigure[]{\label{q263}
\begin{minipage}[b]{0.3\textwidth}
\includegraphics[width=4cm]{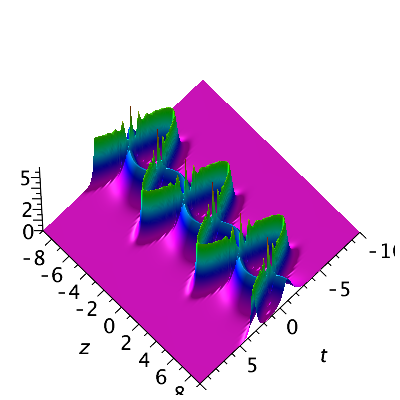}
\end{minipage}
}
\subfigure[]{\label{q271}
\begin{minipage}[b]{0.3\textwidth}
\includegraphics[width=4cm]{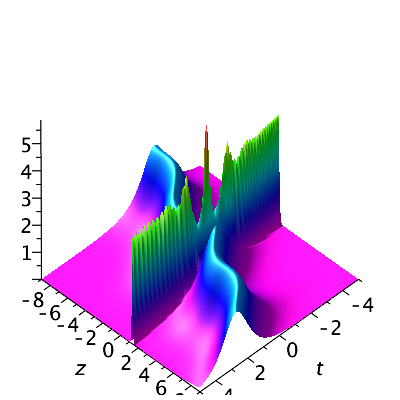} 
\end{minipage}
}
\subfigure[]{\label{q272}
\begin{minipage}[b]{0.3\textwidth}
\includegraphics[width=4cm]{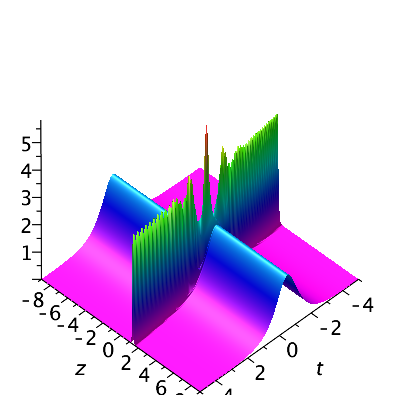}
\end{minipage}
}
\subfigure[]{\label{q273}
\begin{minipage}[b]{0.3\textwidth}
\includegraphics[width=4cm]{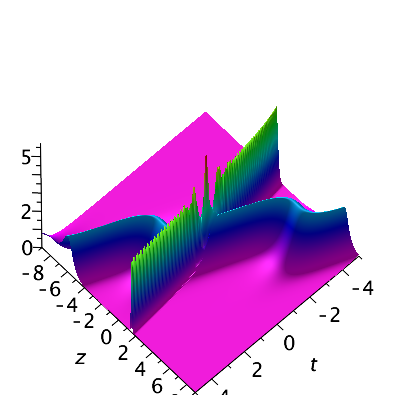}
\end{minipage}
}
\subfigure[]{\label{q281}
\begin{minipage}[b]{0.3\textwidth}
\includegraphics[width=4cm]{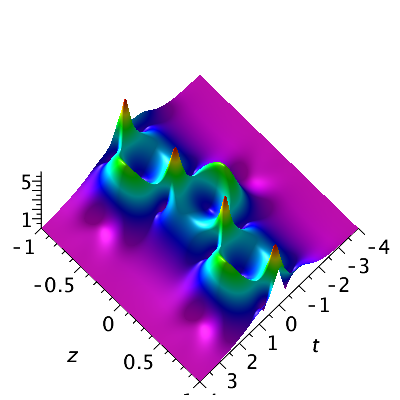}
\end{minipage}
}
\subfigure[]{\label{q291}
\begin{minipage}[b]{0.3\textwidth}
\includegraphics[width=4cm]{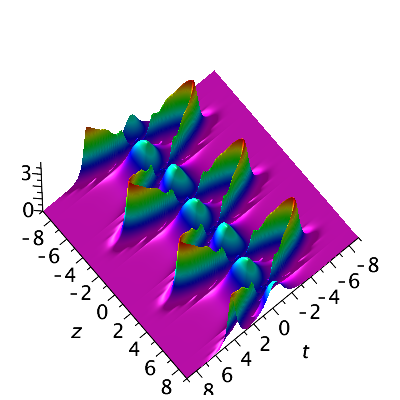}
\end{minipage}
}
\subfigure[]{\label{q2101}
\begin{minipage}[b]{0.3\textwidth}
\includegraphics[width=4cm]{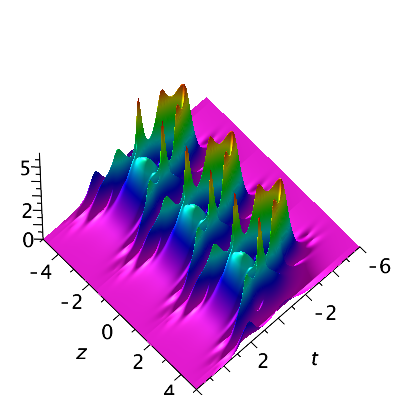}
\end{minipage}
}
\caption{ The evolution plot of 2-soliton solutions in the IVC-Hirota equation: $\alpha_{1}=\alpha_{4}=sin(z)$: ((a) $\delta = 1$; (b) $\delta = \frac{4}{3}$; (c) $\delta = 2$.)  
$\alpha_{1}=\alpha_{4}=1+sin(z)$: ((d) $\delta = 1$; (e) $\delta = \frac{4}{3}$; (f) $\delta = 2$.) (g) $\alpha_{1}=\alpha_{4}=sin(5z)$ and $\delta = 1$; $\alpha_{4}(z) = tan(z)$ and $\delta = 1$: ((h) $\alpha_{1}(z) = sin(z)$; (i)$\alpha_{1}(z) = sin(2z)$.)}
\label{q2sin}
\end{figure}  

When we let $\alpha_{1}(z)=tanh(z)$, Fig.\eqref{q2tanh} shows the dynamic evolution process of the nonlinear term 
$\alpha_{4}$ and $\delta$ with different values. 
\begin{figure}[ht!]
\centering
\subfigure[]{\label{q2111}
\begin{minipage}[b]{0.3\textwidth}
\includegraphics[width=4cm]{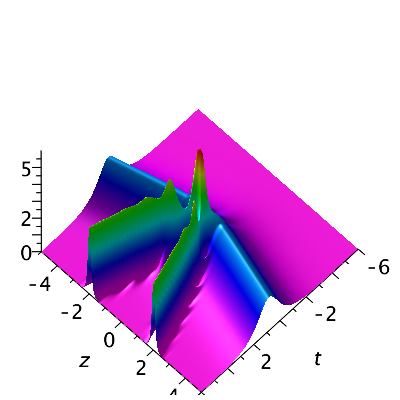} 
\end{minipage}
}
\subfigure[]{\label{q2112}
\begin{minipage}[b]{0.3\textwidth}
\includegraphics[width=4cm]{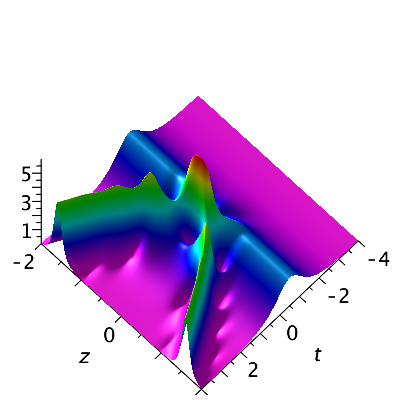}
\end{minipage}
}
\subfigure[]{\label{q2113}
\begin{minipage}[b]{0.3\textwidth}
\includegraphics[width=4cm]{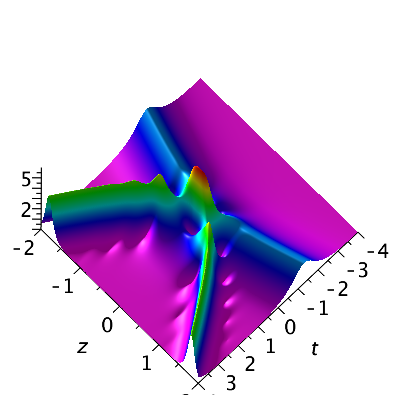}
\end{minipage}
}
\subfigure[]{\label{q2121}
\begin{minipage}[b]{0.3\textwidth}
\includegraphics[width=4cm]{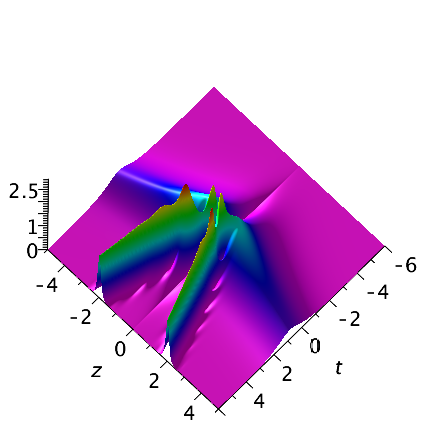} 
\end{minipage}
}
\subfigure[]{\label{q2122}
\begin{minipage}[b]{0.3\textwidth}
\includegraphics[width=4cm]{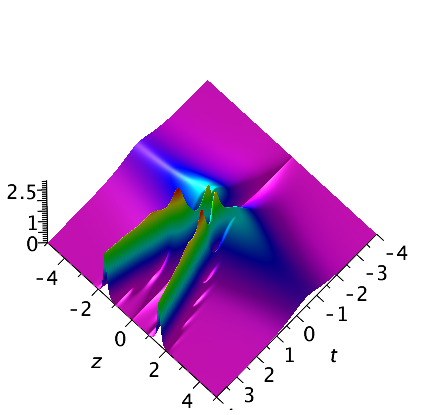}
\end{minipage}
}
\subfigure[]{\label{q2123}
\begin{minipage}[b]{0.3\textwidth}
\includegraphics[width=4cm]{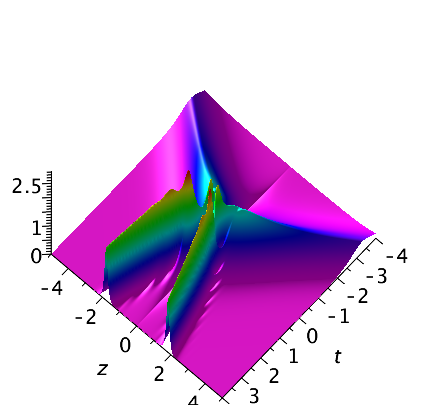}
\end{minipage}
}
\caption{ The evolution plot of 2-soliton solutions in the IVC-Hirota equation: $\alpha_{1}(z) =\alpha_{4}(z) = tanh(z)$: ((a) $\delta = 1$; (b) $\delta = \frac{4}{3}$; (c) $\delta = 2$.) $\alpha_{1}(z) = tanh(z)$ and $\alpha_{4}(z) = z^{2}+1$: ((d) $\delta = 1$; (e) $\delta = \frac{4}{3}$; (f) $\delta = 2$.)}
\label{q2tanh}
\end{figure}  

\section{Soliton matrices for high-order zeros}

We now turn to the high-order zeros in the RH problem of the Hirota equation. For simply, we let functions $P^{+}(\zeta)$ and $P^{-}(\zeta)$ from the above RH problem have only one n-order zero, i.e. $|P^{+}(\zeta)|= (\zeta-\zeta_{1})^{n} \varphi(\zeta)$,
$|P^{-}(\zeta)|=(\zeta-\bar{\zeta}_{1})^{n}\bar{\varphi}(\zeta)$, where $\varphi(\zeta_{ 1} )\neq 0 $ and $\bar{\varphi}(\bar{\zeta}_{ 1} )\neq 0$.

With the help of the idea proposed in \cite{sh2003}, we can consider the elementary zero case under the assumption that the geometric multiplicity of $k_{1}$ and $\bar{\zeta}_{1}$ has the same number. Hence, we need to construct the dressing matrix $\Gamma(\zeta)$ with determinant is  $\frac{(\zeta-\zeta_{1})^{n}}{(\zeta-\bar{\zeta}_{1})^{n}}$. For example, we first consider the elementary zeros with geometric multiplicity 1. In this case, $\Gamma$ is constituted of $n$ elementary dressing factors, i.e.: $\Gamma=\chi_{n}\chi_{n-1} \ldots \chi_{1}, $ where
\[
\begin{array}{l}
\chi_{i}(\zeta)=I+\frac{\bar{\zeta}_{1}-\zeta_{1}}{\zeta-\bar{\zeta}_{1}} P_{i}, P_{i}=\frac{|v_{i}\rangle\langle\bar{v}_{i}|}{\langle\bar{v}_{i} | v_{i}\rangle},|v_{i}\rangle \in \operatorname{Ker}(P_{+} \chi_{1}^{-1} \cdots \chi_{i-1}^{-1}(\zeta_{1})) \\
\end{array}.
\]
In addition, if we let $\hat{P}^{+}(\zeta)=P^{+}(\zeta) \chi_{1}^{-1}(\zeta)$ and $\hat{P}^{-}(\zeta)=\chi_{1}(\zeta) P^{-}(\zeta),$ then it is proved that matrices
$\hat{P}^{+}(\zeta)$ and $\hat{P}^{-}(\zeta)$ are still holomorphic in the respective half plans of $\mathbb{C}$. Moreover, $\zeta_{1}$ and $\bar{\zeta}_{1}$ are still a pair of zeros of $|\hat{P}^{+}(\zeta)|$ and $|\hat{P}^{-}(\zeta)|$, respectively. Thus, $\Gamma(\zeta)^{-1}$ cancels all the high-order zeros for $|P^{+}(\zeta)|$. Moreover, it is necessary to reformulate the dressing factor into summation of fractions, then we derive the soliton matrix $\Gamma(\zeta)$ and its inverse for a pair of an elementary high-order zero. The results can be formulated in the following lemma.

\begin{lemma}
 Consider a pair of elementary high-order zeros of order $n:\left\{\zeta_{1}\right\}$ in $\mathbb{C}_{+}$ and $\left\{\bar{\zeta}_{1}\right\}$ in
$\mathbb{C}_{-}$. Then the corresponding soliton matrix and its inverse can be cast in the following form:
\[\label{s51}
\begin{array}{c}
\Gamma^{-1}(\zeta)=I+\left(\left|p_{1}\right\rangle, \cdots,\left|p_{n}\right\rangle\right) \mathcal{D}(\zeta)\left(\begin{array}{c}
\left\langle q_{n}\right| \\
\vdots \\
\left\langle q_{1}\right|
\end{array}\right), \\
\Gamma(\zeta)=I+\left(\left|\bar{q}_{n}\right\rangle, \cdots,\left|\bar{q}_{1}\right\rangle\right) \bar{\mathcal{D}}(\zeta)\left(\begin{array}{c}
\left\langle\bar{p}_{1}\right| \\
\vdots \\
\left\langle\bar{p}_{n}\right|
\end{array}\right),
\end{array}
\]
where $\mathcal{D}(\zeta)$ and $\bar{\mathcal{D}}(\zeta)$ are $n \times n$ block matrices,
\[
\begin{array}{c}
\mathcal{D}(\zeta)=\left(\begin{array}{ccccc}
(\zeta-\zeta_{1})^{-1} & (\zeta-\zeta_{1})^{-2} & \cdots & (\zeta-\zeta_{1})^{-n} \\
0 & \ddots & \ddots & \vdots \\
\vdots & \ddots & (\zeta-\zeta_{1})^{-1} & (\zeta-\zeta_{1})^{-2} \\
0 & \cdots & 0 & (\zeta-\zeta_{1})^{-1}
\end{array}\right), \quad  \\ \bar{\mathcal{D}}(\zeta)=\left(\begin{array}{cccc}
(\zeta-\zeta_{1})^{-1} & 0 & \cdots & 0 \\
(\zeta-\zeta_{1})^{-2} & (\zeta-\zeta_{1})^{-1} & \ddots & \vdots \\
\vdots & \ddots & \ddots & 0 \\
(\zeta-\zeta_{1})^{-n} & \cdots & (\zeta-\zeta_{1})^{-2} & (\zeta-\zeta_{1})^{-1}
\end{array}\right).
\end{array}
\]
\end{lemma}
This lemma can be proved by induction as in \cite{sh2003}. Besides, we notice that in the expressions for $\Gamma^{-1}(\zeta)$ and $\Gamma(\zeta),$ only half of the vector parameters, i.e.: $\left|p_{1}\right\rangle, \cdots,\left|p_{n}\right\rangle$
 and  $\left\langle \bar{p}_{1}\right|,  \cdots,\left\langle \bar{p}_{n}\right|$ are independent. In fact, the rest of the vector parameters in \eqref{s51} can be derived by calculating the poles of each order in the identity $\Gamma(\zeta) \Gamma^{-1}(\zeta)=I$ at $\zeta=\zeta_{1}$
\[
\Gamma\left(\zeta_{1}\right)\left(\begin{array}{c}
\left|p_{1}\right\rangle \\
\vdots \\
\left|p_{n}\right\rangle
\end{array}\right)=0,
\]
where
\[
\Gamma(\zeta)=\left(\begin{array}{cccc}
\Gamma(\zeta) & 0 & \cdots & 0 \\
\frac{d}{d \zeta} \Gamma(\zeta) & \Gamma(\zeta) & \ddots & \vdots \\
\vdots & \ddots & \ddots & 0 \\
\frac{1}{(n-1) !} \frac{d^{n-1}}{d \zeta^{n-1}} \Gamma(\zeta) & \cdots & \frac{d}{d \zeta} \Gamma(\zeta) & \Gamma(\zeta)
\end{array}\right).
\]
Hence, in terms of the independent vector parameters, results \eqref{s51} can be formulated in a more compact form as in \cite{sh2003} and here we just avoid these overlapped parts. In the following, we derive this compact formula via the method of generalized Darboux transformation (gDT) \cite{llm2}. We intend to investigate the relation between dressing matrices and DT for Hirota equation in the high-order zero case. The essence of the DT is a gauge transformation. Following the scheme proposed in \cite{llm}, we can construct the gDT for Hirota equation as well.

Based on the form of elementary DT\cite{ss}, we can notice $\Gamma_{1}\left(\zeta_{1}+\epsilon\right)\left|v_{1}\left(\zeta_{1}+\epsilon\right)\right\rangle=0.$ 
Furthermore, consider a limitation as follows:
\[
\left|\chi_{1}^{|1|}\left(\zeta_{1}\right)\right\rangle \triangleq \lim _{\epsilon \rightarrow 0} \frac{\Gamma_{1}\left(\zeta_{1}+\epsilon\right)\left|\chi_{1}^{|0|}\left(\zeta_{1}+\epsilon\right)\right\rangle}{\epsilon}=\frac{d}{d \zeta}\left[\Gamma_{1}(\zeta)\left|\chi_{1}^{[0]}(\zeta)\right\rangle\right]_{\zeta=\zeta_{1}},
\]
where $\left|\chi_{1}^{[0]}\left(\zeta_{1}\right)\right\rangle=\left|v_{1}\left(\zeta_{1}\right)\right\rangle.$ Then $\left|\chi_{1}^{(1)}\right\rangle$ can be used to construct the next step DT, i.e.:
\[\Gamma_{1}^{[1]}(\zeta)=\left(I+\frac{\bar{\zeta}_{1}-\zeta_{1}}{\zeta-\bar{\zeta}_{1}} P_{1}^{[1]}\right), \quad P_{1}^{[1]}=\frac{\left|\chi_{1}^{[1]}\right\rangle\left\langle\chi_{1}^{[1]}\right|}{\left\langle\chi_{1}^{[1]} | \chi_{1}^{[1]}\right\rangle}.
\]
The result can be obtained as follows by continuing the above process:
\[
\left|\chi_{1}^{[N]}\right\rangle=\lim _{\epsilon \rightarrow 0} \frac{\Gamma_{1}^{[N-1]} \ldots \Gamma_{1}^{[1]}\Gamma_{1}^{[0]}\left(\zeta_{1}+\epsilon\right)\left|\chi_{1}^{[0]}\left(\zeta_{1}+\epsilon\right)\right\rangle}{\epsilon^{N}}.
\]
The N-times generalized Darboux matrix can be represented as:
\[
T_{N}(\zeta)=\Gamma_{1}^{[N-1]} \ldots \Gamma_{1}^{[1]} \Gamma_{1}^{[0]}(\zeta),
\]
where
\[\Gamma_{1}^{[i]}(\zeta)=\left(I+\frac{\bar{\zeta}_{i}-\zeta_{i}}{\zeta-\bar{\zeta}_{i}} P_{1}^{[i]}\right), \quad P_{1}^{[i]}=\frac{\left|\chi_{1}^{[i]}\right\rangle\left\langle\chi_{1}^{[i]}\right|}{\left\langle\chi_{1}^{[i]} | \chi_{1}^{[i]}\right\rangle}.
\]

In addition, the transformation between different potential matrices is:
\[
Q^{(N)}=Q+i\left[\sigma_{3}, \sum_{j=0}^{N-1}\left(\bar{\zeta_{1}}-\zeta_{1}\right) P_{1}^{[j]}\right].
\]
In this expression, $P_{1}^{[i]}$ is rank-one matrices, so $\Gamma_{1}^{[i]}(\zeta)$ can be also decomposed into the summation of simple fraction, that means the multiple product form of $T_{N}$ can be directly simplified by the conclusion of Lemma $1$. In other words, the above generalized Darboux matrix for Hirota equation can be given in the following theorem:
\begin{theorem}
In the case of one pair of elementary high-order zero, the generalized Darboux matrix for Hirota equation can be represented as \cite{ss}:
\[
T_{N}=I-Y M^{-1} \bar{\mathcal{D}}(\zeta) Y^{\dagger},
\]
where $\bar{\mathcal{D}}(\zeta)$ is $N \times N$ block Toeplitz matrix which has been given before, $Y$ is a $2 \times N$ matrix:
\[
\begin{array}{c}
Y=\left(\left|v_{1}\right\rangle, \ldots, \frac{\left|v_{1}\right\rangle^{(N-1)}}{(N-1) !}\right), \\
\left|v_{1}\right\rangle^{(j)}=\lim_{\epsilon \rightarrow 0} \frac{d^{j}}{d \epsilon^{j}}\left|v_{1}\left(\zeta_{1}+\epsilon\right)\right\rangle,
\end{array}
\]
and $M$ is $N \times N$ matrix:
\[
M=\left
(\begin{array}{ll}
M_{j,k}^{[m,n]}
\end{array}\right)_{N \times N}
\]
with
\[
 M_{j,k}^{[m,n]}=\lim _{\epsilon, \bar{\epsilon} \rightarrow 0} \frac{1}{(m-1) !(n-1) !} \frac{\partial^{m-1}}{\partial \epsilon^{m-1}} \frac{\partial^{n-1}}{\partial(\bar{\epsilon})^{n-1}}\left[\frac{\left\langle v_{j} | v_{k}\right\rangle}{\zeta_{j}-\bar{\zeta}_{k}+\epsilon-\bar{\epsilon}}\right].
\]
\end{theorem}
Theorem 1 can be proved via directly calculation as in \cite{llm}. Therefore, if $\Phi^{|N|}=T_{N} \Phi$, then $\Phi^{[N]}$ indeed solves spectral problem \eqref{b}.
Substituting $T_{N}$ into the above relation and letting spectral $\zeta$ go to infinity, we have the relation:
\[
Q^{[N]}=Q-i\left[\sigma_{3},\left(\left|v_{1}\right\rangle, \ldots, \frac{\left|v_{1}\right\rangle^{(N-1)}}{(N-1) !}\right) M^{-1}\left(\begin{array}{c}
\left\langle v_{1}\right| \\
\vdots \\
\frac{\langle v_{1}|^{(N-1)}}{(N-1)!}
\end{array}\right)\right].
\]
Moreover, the transformations between potential functions are:
\[\label{59}
Q_{j, l}^{[N]}=Q_{j, l}^{[0]}+2i\frac{|A_{j, l}|}{|M|}, \quad A_{j, l}=\left[\begin{array}{cc}
M & Y[l]^{\dagger} \\
Y[j] & 0
\end{array}\right], 1 \leq j, l \leq 2.
\]
Here the subscript $_{j, l}$ denotes the $j$ th row and $l$ th column element of matrix $A$, and $Y[l]$ represents the $j$ th row of matrix $Y$.

\section{High-order soliton solution for the Hirota and IVC-equation}

Firstly, choice a single pair of purely imaginary eigenvalues, $\zeta_{1}=i \eta_{1} \in i \mathbb{R}_{+},$ and $\bar{\zeta}_{1}=i \bar{\eta}_{1} \in i \mathbb{R}_{-},$ where $\eta_{1}>$
0 and $\bar{\eta}_{1}=-\eta_{1}<0$ to get a brevity second-order fundamental soliton expression.  In this case, taking $v_{10}(\epsilon)=\left[1, \mathrm{e}^{i \theta_{10}-\theta_{11}\epsilon}\right]^{\mathrm{T}}$ and $\bar{v}_{10}(\bar{\epsilon})=\left[1, \mathrm{e}^{i \bar{\theta}_{10}-\tilde{\theta}_{11} \bar{\epsilon}}\right]^{\mathrm{T}},$ where $\theta_{10}, \theta_{11}, \bar{\theta}_{10}, \bar{\theta}_{11}$ are real constants. Substituting these expressions into high-order soliton formula \eqref{59} with  $N=2$, $Q_{1,2}^{[0]}=0$, then  the analytic expression for the second-order  soliton solution $u^{[2]}(T, Z)$ is obtained as follows.
\begin{equation}\label{b29}
\begin{aligned}
2(\bar{\eta}_{1}-\eta_{1})\frac{t_{11}e^{2\bar{\eta}_{1}T+(4i\alpha\bar{\eta}_{1}^{2}-8\beta\bar{\eta}_{1}^3)Z+i\bar{\theta_{10}}}
+t_{12}e^{2\eta_{1} T+(4i\alpha\eta_{1}^{2}-8\beta\eta_{1}^3)Z-i\theta_{10}}}{4 \cosh ^{2}((\eta_{1}-\bar{\eta}_{1})T+4\beta(\bar{\eta}_{1}^{3}-\eta_{1}^{3})Z+2i\alpha(\eta_{1}^{2}-\bar{\eta_{1}}^{2})Z-\frac{i}{2}(\theta_{10}+\bar{\theta}_{10}))+F(T, Z)},
\end{aligned}
\end{equation}
\begin{equation}\notag
\begin{split}
&t_{11}=(\bar{\eta}_{1}-\eta_{1})(-24\eta_{1}^{2}\beta Z+8i\eta_{1}\alpha Z+2T+i\theta_{11})-2,\\
&t_{12}=(\eta_{1}-\bar{\eta}_{1})(-24\bar{\eta}_{1}^{2}\beta Z+8i\bar{\eta}_{1}\alpha Z+2T-i\bar{\theta}_{11})-2,\\
&F(T,Z)=-(t_{11}+2)(t_{12}+2).
\end{split}
\end{equation}
The second-order soliton solution $u^{[2]}(T, Z)$ combines exponential functions with algebraic polynomials, contains six real parameters: $\eta_{1}, \bar{\eta}_{1}, \theta_{10}, \bar{\theta}_{10}, \theta_{11},$ and $\bar{\theta}_{11}$. The center trajectory $\Sigma_{+}$ and $\Sigma_{-}$ for this solution can be approximatively described by the following two curves:
\[
\begin{aligned}
\Sigma_{+}: (\eta_{1}-\bar{\eta}_{1})T+4\beta(\bar{\eta}_{1}^{3}-\eta_{1}^{3})Z+\frac{1}{2}ln|F|=0,\\
\Sigma_{-}: (\eta_{1}-\bar{\eta}_{1})T+4\beta(\bar{\eta}_{1}^{3}-\eta_{1}^{3})Z-\frac{1}{2}ln|F|=0.
\end{aligned}
\]

Moreover, regardless of the effect brought by the logarithmic part when $Z \rightarrow \pm\infty$, two solitons separately move along each curve in a nearly same velocity, which is approximate to \[ \begin{aligned}
V=-4\beta (\eta_{1}^{2}+\eta_{1}\bar{\eta}_{1}+\eta_{1}^{2}).
\end{aligned} \]
Due to $\eta_{1}-\bar{\eta}_{1}>0$, with simple calculation, it is found that $|u^{[2]}(T, Z)|$ possesses the following asymptotic estimation:
\begin{equation}
|u^{[2]}(T, Z)|\rightarrow 0, \quad |T| \rightarrow \pm \infty.
\end{equation}
However, with the development of time, a simple asymptotic analysis with estimation on the leading-order terms shows that: when soliton \eqref{b29} is moving on $\Sigma_{+}$ or $\Sigma_{-}$, its amplitudes $|u^{[2]}(T, Z)|$ can approximately vary as
\begin{equation}
 |u^{[2]}(T, Z)|  \sim
\begin{cases}
\frac{2\left|\eta_{1}-\bar{\eta}_{1}\right| \mathrm{e}^{\left(\eta_{1}+\bar{\eta}_{1}\right)T}}{\left|\mathrm{e}^{4i\alpha(\eta_{1}^{2}-\bar{\eta}_{1}^{2})Z-i (\operatorname{arg}[\mathcal{F}(T, Z)]+2k\pi)+ i\left(\theta_{10}+\bar{\theta}_{10}\right)}+1\right|}, \quad Z \sim  +\infty,\\
\\
\frac{2\left|\eta_{1}-\bar{\eta}_{1}\right| \mathrm{e}^{-(\eta_{1}+\bar{\eta}_{1})T}}{\left|\mathrm{e}^{-4i (\alpha(\eta_{1}^{2}-\bar{\eta}_{1}^{2})Z-i (\operatorname{arg}[\mathcal{F}(T, Z)]+2k\pi)- i\left(\theta_{10}+\bar{\theta}_{10}\right)}+1\right|}, \quad Z \sim -\infty,
\end{cases}
\end{equation}
where $k\in \Z$.

 Letting $\eta_{1}= \frac{i}{2}$, $\bar{\eta}_{1}=-\frac{i}{2}$ and $\theta_{10}=\bar{\theta}_{10}=\theta_{11}=\bar{\theta}_{11}=0$, the value of $(\alpha,\beta)$ will change the velocity, direction and shape of the soliton figure.  We can divide the analysis into three cases three cases as shown in  
Figs.\eqref{01gjgz}, \eqref{11gjgz} and \eqref{10gjgz}: the first case is that $\alpha = 0$ and $\beta = 1$;  the second case is $\alpha = 1$ and $\beta =1$;  the third case is $\alpha = 0$ and $\beta = 1$.  Graphically, the soliton evolution of the Hirota equation is more similar to that of the KdV equation. That is, the value of higher term coefficient $\beta$ plays a decisive role in dynamic analysis.
\begin{figure}[ht!]
\centering
\subfigure[]{\label{01gjgz}
\begin{minipage}[b]{0.2\textwidth}
\includegraphics[width=4cm]{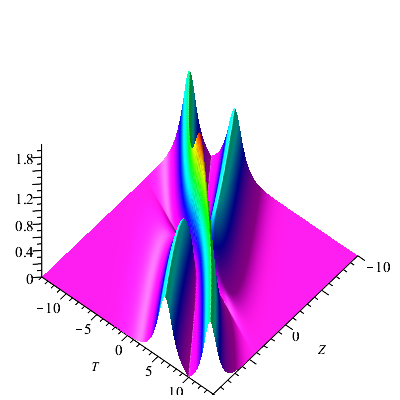} 
\end{minipage}
}
\subfigure[]{\label{11gjgz}
\begin{minipage}[b]{0.2\textwidth}
\includegraphics[width=4cm]{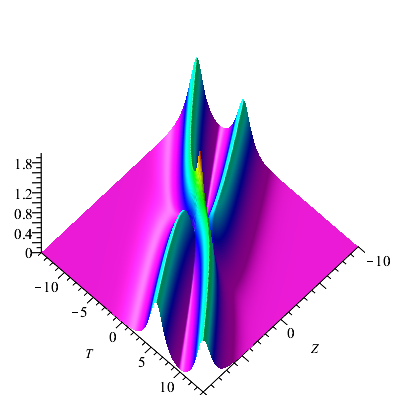} 
\end{minipage}
}
\subfigure[]{\label{10gjgz}
\begin{minipage}[b]{0.2\textwidth}
\includegraphics[width=4cm]{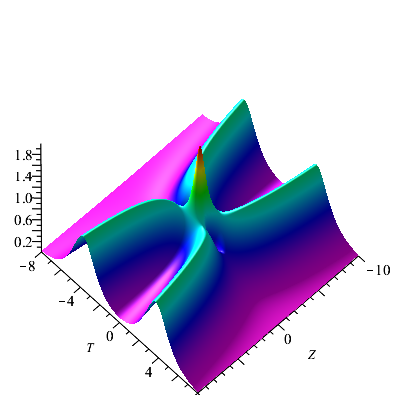} 
\end{minipage}
}
\caption{ The evolution plot of the second-order soliton solutions in the Hirota equation: (a) $\alpha=0$ and $\beta=1$; (b) $\alpha=1$ and $\beta=1$; (c) $\alpha=1$ and $\beta=1$. }
\label{gj2gz}
\end{figure}

Using the explicit transformation \eqref{tr}, we can obtain abundance of high-order soliton solutions $q^{[N]}(T, Z)$ of the IVC-Hirota equation from the known solutions 
$u^{[N]}(T, Z)$ of the Hirota equation. Now taking $N=2$, $Q_{1,2}^{[0]}=\theta_{10}=\bar{\theta}_{10}=\theta_{11}=\bar{\theta}_{11}=0$, $\zeta_{1}= 1+i$ and $\zeta_{2}=1-i$, without loss of generality, we present below the dynamic evolution analysis of the second-order soliton solution of the variable coefficient equation when $\alpha=0$. When $\alpha_{1}(z) =\alpha_{4}(z) = z$, the solitons are symmetric about the line $z=0$ and have only one crest. Comparing with Fig.\eqref{gq11}, we can find that the dynamics of solution $q^{2}$ in Fig.\eqref{gq12} is different that the symmetric is moved to $z=-1$ and the peaks change from one to two when $\alpha_{1}(z) = \alpha_{4}(z) = 1+z$. When $\alpha_{1}(z) = z$ and $\alpha_{4}(z) = z^{2}+1$, the maximum amplitude of the second-order soliton solution appears at the position of interaction of soliton, and the dynamic behaviour of the other positions is similar to that of the corresponding $1$-soliton solution. Similarly, we also consider the cases of 
$\alpha_{1}(z)=\alpha_{4}(z)=z^{2}$ and $\alpha_{1}(z) = z^{2}$, $\alpha_{4}(z) = 1+z^{2}$. The detailed dynamic behaviour of the solution can be observed in Fig. \eqref{zzzzz}.
\begin{figure}[ht!]
\centering
\subfigure[]{\label{gq11}
\begin{minipage}[b]{0.18\textwidth}
\includegraphics[width=3.2cm]{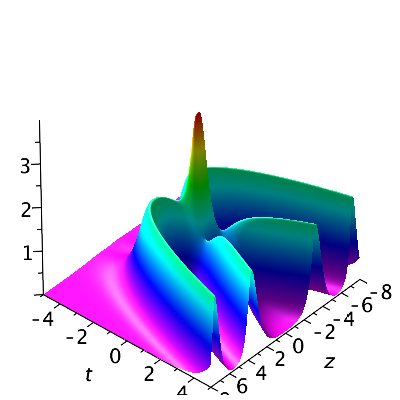} \\
\end{minipage}
}
\subfigure[]{\label{gq12}
\begin{minipage}[b]{0.18\textwidth}
\includegraphics[width=3.2cm]{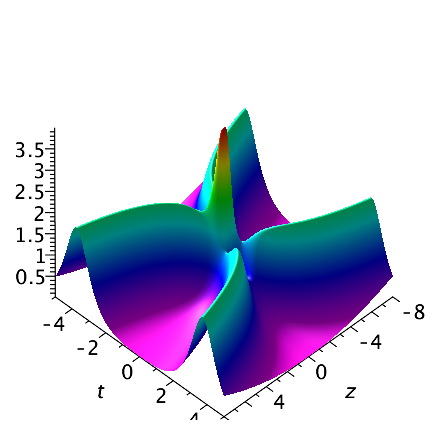} \\
\end{minipage}
}
\subfigure[]{\label{gq21}
\begin{minipage}[b]{0.18\textwidth}
\includegraphics[width=3.2cm]{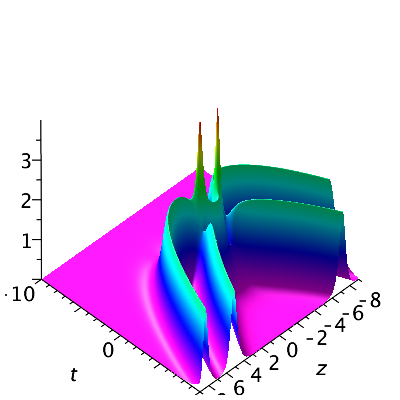} \\
\end{minipage}
}
\subfigure[]{\label{gq22}
\begin{minipage}[b]{0.18\textwidth}
\includegraphics[width=3.2cm]{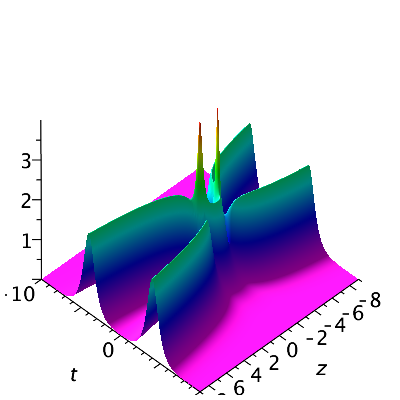} \\
\end{minipage}
}
\subfigure[]{\label{gq31}
\begin{minipage}[b]{0.18\textwidth}
\includegraphics[width=3.2cm]{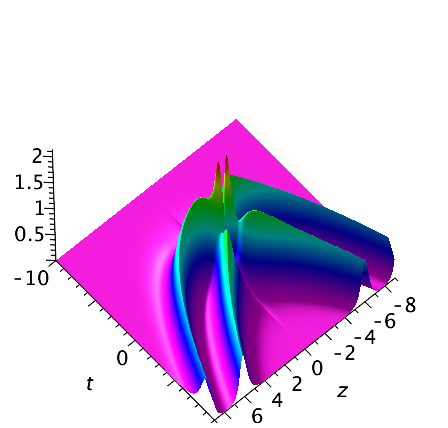} \\
\end{minipage}
}
\subfigure[]{\label{gq32}
\begin{minipage}[b]{0.18\textwidth}
\includegraphics[width=3.2cm]{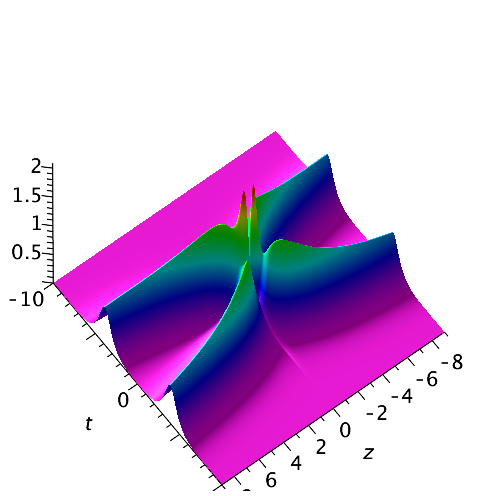} \\
\end{minipage}
}
\subfigure[]{\label{gq41}
\begin{minipage}[b]{0.18\textwidth}
\includegraphics[width=3.2cm]{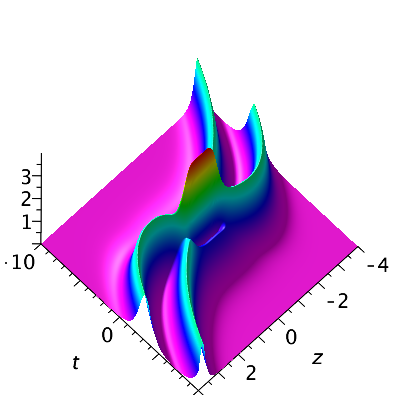} \\
\end{minipage}
}
\subfigure[]{\label{gq42}
\begin{minipage}[b]{0.18\textwidth}
\includegraphics[width=3.2cm]{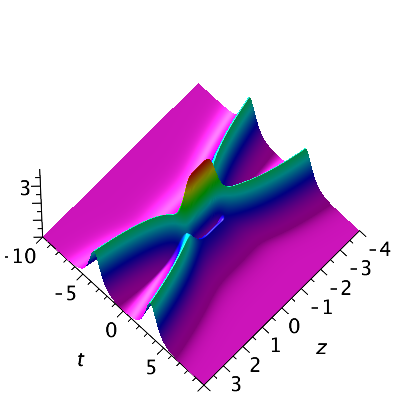} \\
\end{minipage}
}
\subfigure[]{\label{gq51}
\begin{minipage}[b]{0.18\textwidth}
\includegraphics[width=3.2cm]{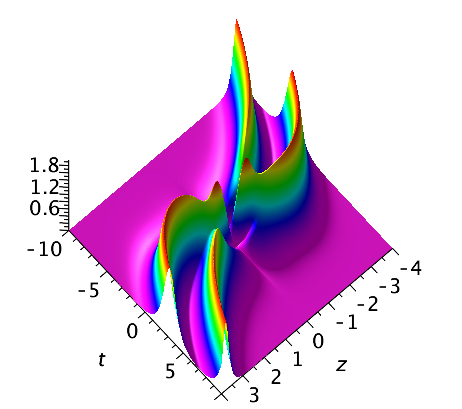} \\
\end{minipage}
}
\subfigure[]{\label{gq52}
\begin{minipage}[b]{0.18\textwidth}
\includegraphics[width=3.2cm]{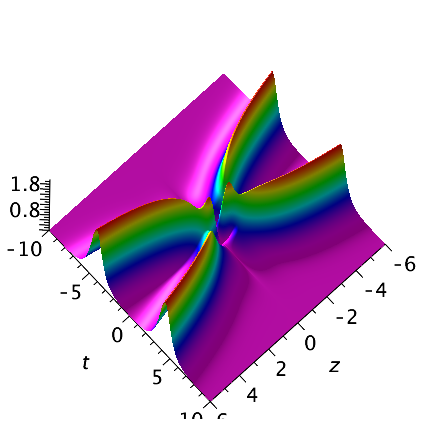} \\
\end{minipage}
}
\caption{The evolution plot of the second-order soliton solutions in the IVC-Hirota equation: $\alpha_{1}(z) =\alpha_{4}(z) = z$: ((a) $\delta=1$; (b) $\delta=\frac{4}{3}$.)
$\alpha_{1}(z)=\alpha_{4}(z)= 1+z$: ((c) $\delta=1$; (d) $\delta=\frac{4}{3}$.) $\alpha_{1}(z)=z$ and $\alpha_{4}(z) = z^{2}+1$: ((e) $\delta=1$; (f) $\delta=\frac{4}{3}$.)
 $\alpha_{1}(z)=\alpha_{4}(z)=z^{2}$: ((g) $\delta=1$; (h) $\delta=\frac{4}{3}$.) $\alpha_{1}(z) = z^{2}$ and $\alpha_{4}(z) = 1+z^{2}$: ((i) $\delta=1$; (j) $\delta=\frac{4}{3}$.)
   }
\label{zzzzz}
\end{figure}

When we take trigonometric functions as the excitations function, we can obtain very rich non-singular convergent second-order solutions which are shown in Fig.\eqref{gqsin} by adjusting the parameters. For example, when $\alpha_{1}(z) =\alpha_{4}(z) =10sin(z)$, we can construct heart-shaped periodic waves when $\delta = 1$ and $\delta = 2$ (see Figs.\eqref{gq141} and \eqref{gq143}). For $\delta =\frac{4}{3}$, an O-shaped periodic wave is plotted in Fig.\eqref{gq142}. When $\alpha_{1}(z) = \alpha_{4}(z) = sin(5z)$ or $\alpha_{1}(z)=\alpha_{4}(z) = sin(3z)$, we can see the dynamic behaviours in Figs.\eqref{gq91} and \eqref{gq1112} which are similar to the breather solution, there are a peak and two troughs in each periodic.
\begin{figure}[ht]
\centering
\subfigure[]{\label{gq71}
\begin{minipage}[b]{0.18\textwidth}
\includegraphics[width=3.2cm]{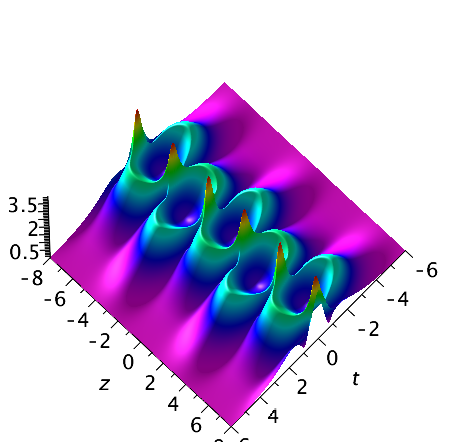} \\
\end{minipage}
}
\subfigure[]{\label{gq72}
\begin{minipage}[b]{0.18\textwidth}
\includegraphics[width=3.2cm]{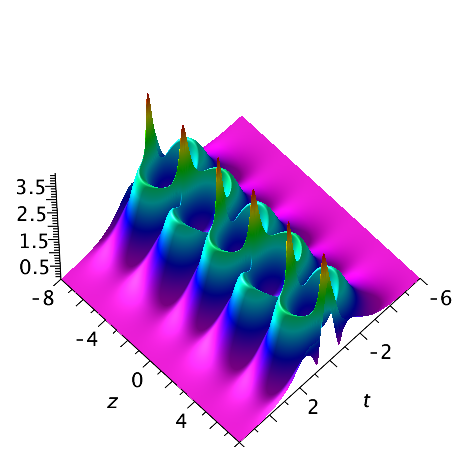} \\
\end{minipage}
}
\subfigure[]{\label{gq73}
\begin{minipage}[b]{0.18\textwidth}
\includegraphics[width=3.2cm]{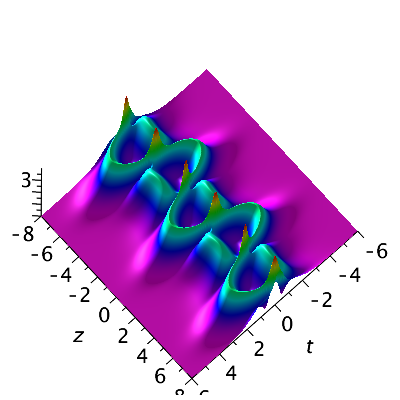} \\
\end{minipage}
}
\subfigure[]{\label{gq141}
\begin{minipage}[b]{0.18\textwidth}
\includegraphics[width=3.2cm]{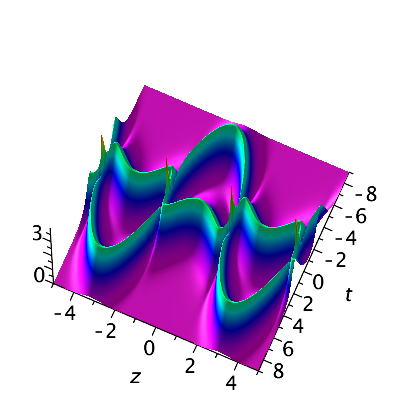} \\
\end{minipage}
}
\subfigure[]{\label{gq142}
\begin{minipage}[b]{0.18\textwidth}
\includegraphics[width=3.2cm]{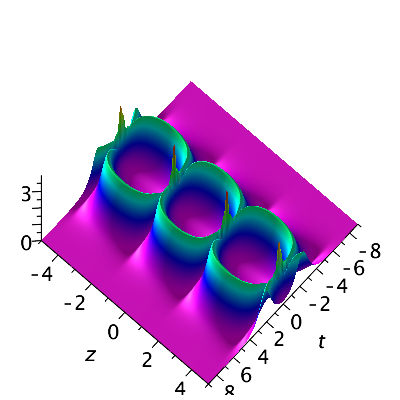} \\
\end{minipage}
}
\subfigure[]{\label{gq143}
\begin{minipage}[b]{0.18\textwidth}
\includegraphics[width=3.2cm]{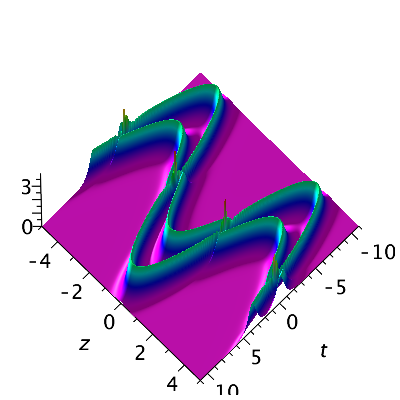} \\
\end{minipage}
}
\subfigure[]{\label{gq81}
\begin{minipage}[b]{0.18\textwidth}
\includegraphics[width=3.2cm]{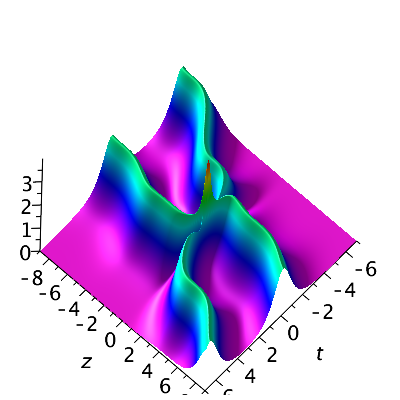} \\
\end{minipage}
}
\subfigure[]{\label{gq82}
\begin{minipage}[b]{0.18\textwidth}
\includegraphics[width=3.2cm]{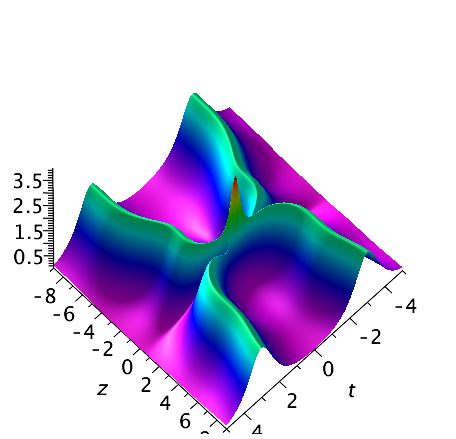} \\
\end{minipage}
}
\subfigure[]{\label{gq83}
\begin{minipage}[b]{0.18\textwidth}
\includegraphics[width=3.2cm]{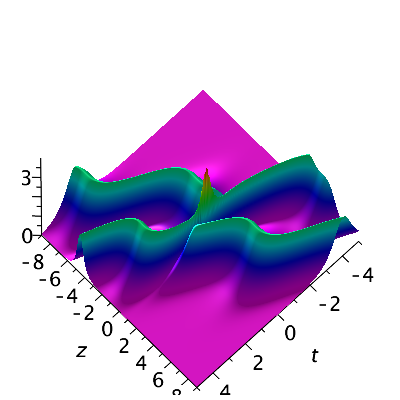} \\
\end{minipage}
}
\subfigure[]{\label{gq101}
\begin{minipage}[b]{0.18\textwidth}
\includegraphics[width=3.2cm]{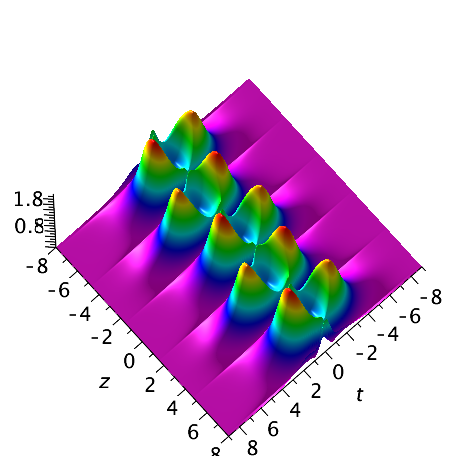} \\
\end{minipage}
}
\subfigure[]{\label{gq102}
\begin{minipage}[b]{0.18\textwidth}
\includegraphics[width=3.2cm]{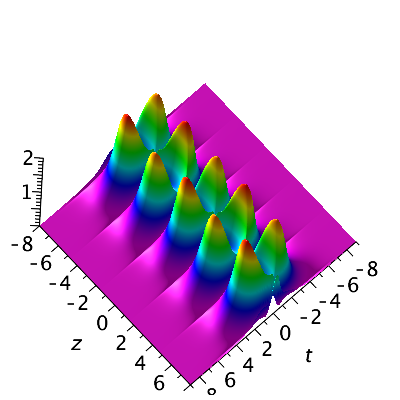} \\
\end{minipage}
}
\subfigure[]{\label{gq103}
\begin{minipage}[b]{0.18\textwidth}
\includegraphics[width=3.2cm]{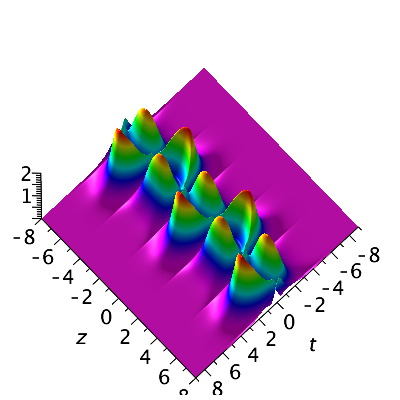} \\
\end{minipage}
}
\subfigure[]{\label{gq91}
\begin{minipage}[b]{0.18\textwidth}
\includegraphics[width=3.2cm]{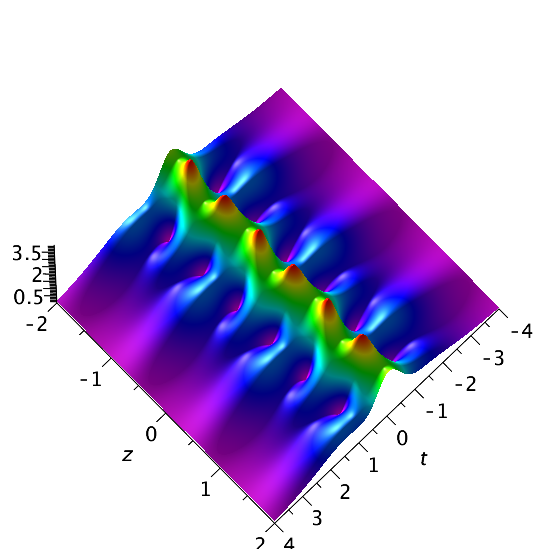} \\
\end{minipage}
}
\subfigure[]{\label{gq111}
\begin{minipage}[b]{0.18\textwidth}
\includegraphics[width=3.2cm]{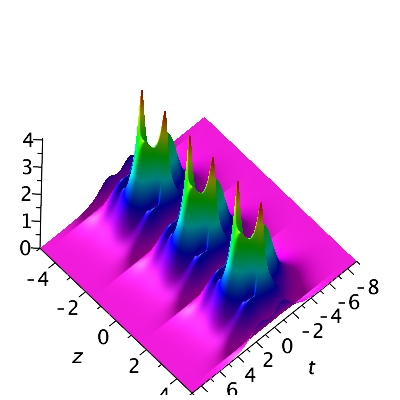} \\
\end{minipage}
}
\subfigure[]{\label{gq1112}
\begin{minipage}[b]{0.18\textwidth}
\includegraphics[width=3.2cm]{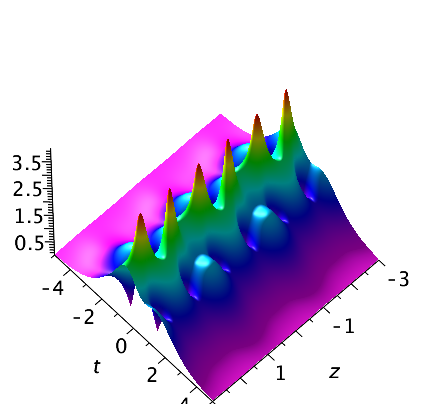} \\
\end{minipage}
}
\caption{The evolution plot of the second-order soliton solutions in the IVC-Hirota equation: $\alpha_{1}(z) = \alpha_{4}(z) = sin(z)$: ((a) $\delta = 1$; (b) $\delta =\frac{4}{3}$; (c) $\delta = 2$.) 
$\alpha_{1}(z) =\alpha_{4}(z) =10sin(z)$: ((d) $\delta = 1$; (e) $\delta =\frac{4}{3}$; (f) $\delta = 2$.)
$\alpha_{1}(z)=\alpha_{4}(z) =1+ sin(z)$: ((g) $\delta = 1$; (h) $\delta =\frac{4}{3}$; (i) $\delta = 2$.)
$\alpha_{1}(z) = sin(z)$ and $\alpha_{4}(z) = tan(z)$: ((j) $\delta = 1$; (k) $\delta =\frac{4}{3}$; (l) $\delta = 2$.)
(m) $\alpha_{1}(z) = \alpha_{4}(z) = sin(5z)$ and $\delta = 1$;
(n) $\alpha_{1}(z) = sin(2z)$, $\alpha_{4}(z) = tan(z)$ and $\delta = 1$; (o) $\alpha_{1}(z) = \alpha_{4}(z) = sin(3z)$ and $\delta = 1.$}
\label{gqsin}
\end{figure}

In the last, taking $\alpha_{1}(z) = tanh(z)$, in order to get meaningful nonsingular convergent solutions we can let $\alpha_{4}(z) =tanh(z)$ or $\alpha_{4}(z) =z^{2}+1$, Fig.\eqref{gq2tanh} shows the dynamic evolution process of the second-order soliton solution with different values of the parameter $\delta$. 
\begin{figure}[ht!]
\centering
\subfigure[]{\label{gq121}
\begin{minipage}[b]{0.25\textwidth}
\includegraphics[width=3.4cm]{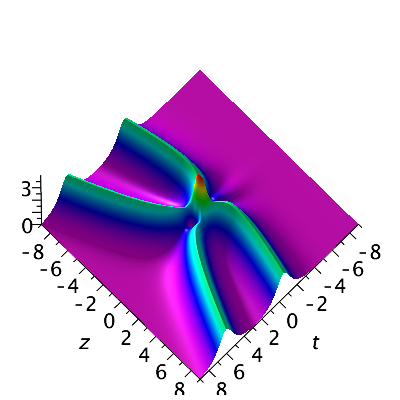} \\
\end{minipage}
}
\subfigure[]{\label{gq122}
\begin{minipage}[b]{0.25\textwidth}
\includegraphics[width=3.4cm]{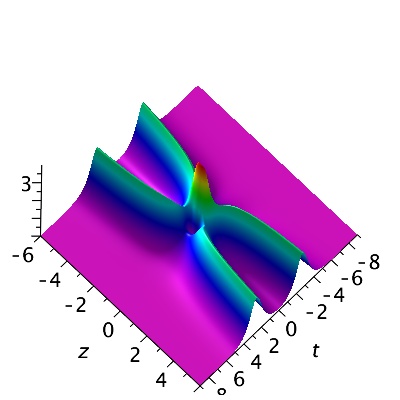} \\
\end{minipage}
}
\subfigure[]{\label{gq123}
\begin{minipage}[b]{0.25\textwidth}
\includegraphics[width=3.4cm]{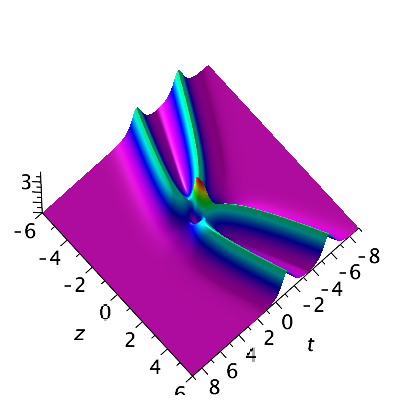} \\
\end{minipage}
}
\subfigure[]{\label{gq131}
\begin{minipage}[b]{0.25\textwidth}
\includegraphics[width=3.4cm]{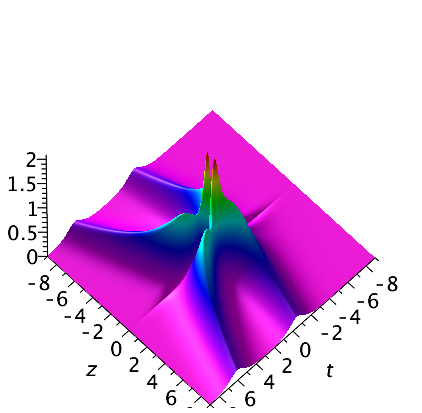} \\
\end{minipage}
}
\subfigure[]{\label{gq132}
\begin{minipage}[b]{0.25\textwidth}
\includegraphics[width=3.4cm]{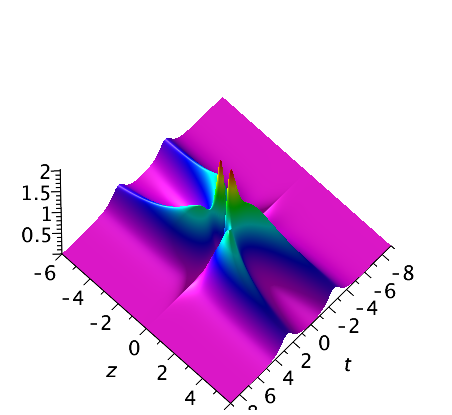} \\
\end{minipage}
}
\subfigure[]{\label{gq133}
\begin{minipage}[b]{0.25\textwidth}
\includegraphics[width=3.4cm]{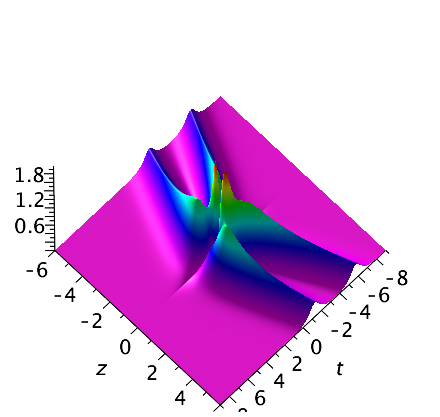} \\
\end{minipage}
}
\caption{The evolution plot of the second-order soliton solutions in the IVC-Hirota equation: $\alpha_{1}(z) =\alpha_{4}(z) = tanh(z)$: ((a) $\delta = 1$; (b) $\delta =\frac{4}{3}$; (c) $\delta = 2$.)
 $\alpha_{1}(z) =tanh(z)$ and $\alpha_{4}(z) =z^{2}+1$: ((d) $\delta = 1$; (e) $\delta =\frac{4}{3}$; (f) $\delta = 2$.)
}
\label{gq2tanh}
\end{figure}  

\section{Conclusion and discussion}

In summary, many new soliton solutions for the IVC-Hirota equation are implemented by using the RH method and a special transformation relationship.
First, the soliton matrices are constructed by studying the corresponding RH problem. By regularizing the RH problem with simple zeros, we get the general N-soliton formula for the Hirota equation. In addition, the high-order soliton matrices are also obtained by considering the multiple zeros of the RH problem. Then the $N$-soliton matrix and high-order soliton matrices of the IVC-Hirota equation are presented from the correspond soliton matrix of generalized Hirota equation by a special transformation relationship. We find when we let second-order term coefficient $\alpha$ is equal to 0 in the transformation relationship, third-order term coefficient $\beta$  disappear from the solution of the IVC-Hirota equation. Namely, the high-order term coefficient $\beta$ has no influence on the solution of IVC-Hirota equation which is obtained by the special transformation relationship with $\alpha=0$. 

The 2-soliton collision dynamics, the asymptotic behavior of the 2-soliton and the long time asymptotic estimates for the high-order one-soliton solution of Hirota eqution are detailed in this paper. For Hirota equation, compared with the second-order coefficient $\alpha$, the value of the higher-order coefficient $\beta$ plays a more important role in the dynamic behavior of the solution.  In addition, a detailed dynamic analysis is given for the  N-soliton solution and the high-order soliton solution of the IVC-Hirota equation. Most notably, by analyzing the dynamics of the N-solitons and high-order solitons of the IVC-Hirota equation, we have found many new waveforms that have never been reported before, which are very important in theory and practice. For example, when both nonlinear effect and dispersion effect are taken as periodic functions, interesting new waves such as heart-shaped periodic wave and O-shaped periodic wave can be constructed by adjusting the parameters.

Following the work in this paper, we plan to use numerical methods in future work to simulate high-order soliton solutions of integrable equations with non-zero boundaries.


\end{document}